\documentclass[pdflatex,sn-mathphys-ay]{sn-jnl}

\usepackage[utf8]{inputenc}
\usepackage[T1]{fontenc}
\usepackage[english]{babel}
\usepackage{comment}
\usepackage{graphicx}%
\usepackage{tabularx}
\usepackage{amsmath,amssymb,amsfonts}%
\usepackage{amsthm}%
\usepackage{mathrsfs}%
\usepackage[title]{appendix}%
\usepackage{xcolor}%
\usepackage{textcomp}%
\usepackage{manyfoot}%
\usepackage{booktabs}%
\usepackage{hyperref,float,array,multirow}
\usepackage{bbold} 
\usepackage{colortbl}
\usepackage{makecell}
\usepackage{booktabs}
\usepackage{float}
\usepackage{hhline}
\usepackage{makeidx,geometry}
\usepackage{listings}%
\usepackage{algorithm,algorithmic}

\newtheorem{rem}{Remark}

\newtheorem{theo}{Theorem}
\newtheorem{cor}{Corollary}
\newtheorem{lem}{Lemma}

\newtheorem{assump}{Assumption}
\newtheorem{nota}{Notation}

\newcommand{\inas}{\overset{\scriptstyle a.s.}{\longrightarrow}}

\newcommand{\inprob}{\overset{\scriptstyle \mathbb{P}}{\longrightarrow}}

\theoremstyle{thmstyleone}%
\newtheorem{theorem}{Theorem}
\newtheorem{proposition}[theorem]{Proposition}%

\theoremstyle{thmstyletwo}%

\theoremstyle{thmstylethree}%

\begin{document}
	
	\title[Article Title]{Existence and consistency of weighted maximum
		likelihood estimator for extreme quantile regression.}
	
	
	\author*[1]{\fnm{Lucien M. } \sur{Vidagbandji}}\email{mahutin-lucien.vidagbandji@univ-lehavre.fr}
	\author[1]{\fnm{Alexandre} \sur{Berred}}\email{alexandre.berred@univ-lehavre.fr}
%
	\author[2]{\fnm{Cyrille} \sur{Bertelle}}\email{cyrille.bertelle@univ-lehavre.fr}
	\author[2]{\fnm{Laurent} \sur{Amanton}}\email{laurent.amanton@univ-lehavre.fr}
%
%

	\abstract{Estimating extreme conditional quantiles faces two major challenges: understanding complex nonlinear relationships between variables and accurately extrapolating into the tails of the distribution, where data are sparse. To address both issues, we introduce a novel approach that combines the theoretical framework of block maxima with the predictive power of generalized random forests. The conditional distribution of maxima is modeled by a generalized extreme value (GEV) distribution whose parameters depend on the covariates and are estimated via a weighted maximum likelihood procedure, with weights derived from generalized random forests to capture complex high-dimensional structures. The conditional quantile estimator follows from the inversion of the estimated GEV distribution function. We establish theoretical results ensuring the existence and consistency of the proposed weighted maximum likelihood estimator. The method provides a robust and flexible framework for extreme quantile regression. The performance of the proposed approach is illustrated using simulated data, as well as an application to financial portfolio losses from stocks listed on the NYSE, AMEX, and NASDAQ.
	}
	

	\maketitle
 \section{Introduction }
Quantile regression is a statistical method based on the estimation of conditional quantiles. More precisely, given a response variable $Y \in \mathcal{Y} \subset \mathbb{R}$ depending on a covariate $X \in \mathcal{X} \subset \mathbb{R}^p$, the objective of quantile regression is to estimate the conditional quantile $\mathcal{Q}_{\tau}(Y \mid X = x)$ of the response variable $Y$ given $X = x$, for any $\tau \in (0,1)$ (\cite{koenker_regression_1978}).   Although this method provides more comprehensive information on the distribution of $Y$ as a function of the predictor variable $X$ than the conditional mean alone, it faces two main challenges. The first arises when the relationship between the quantile function $\mathcal{Q}_{\tau}(Y \mid X = x)$ and the covariate $x$ is complex and highly nonlinear, and when the dimension $p$ of the predictor space $\mathcal{X}$ is large. The second difficulty arises when we want to estimate a conditional quantile corresponding to an extreme probability level, that is, a quantile of level $\tau_n \to 1$ such that $n(1 - \tau_n)$ remains finite when $n \to \infty$, where $n$ denotes the size of the available sample. This situation poses a real challenge, as the estimation requires extrapolation in the tail of the distribution, which is often problematic due to the scarcity of data in this region. In this paper, we propose a method for estimating the conditional quantile that simultaneously addresses these two issues.

It should be noted that the methods proposed in the literature for estimating extreme conditional quantiles can broadly be classified into several groups
(\cite{wang_estimation_2013}).  
The first consists in directly applying quantile regression to model extreme conditional quantiles, without imposing any assumptions on the conditional distribution $Y \mid X = x$. However, this approach suffers from a lack of reliability when extrapolating in the tail of the distribution, which corresponds to an area of low data density (\cite{chernozhukov_extremal_2005}, \cite{bremnes_probabilistic_2004}, \cite{hallock_quantile_2001}).  The second extends extreme value theory to the regression framework by adopting a local estimation method, based on neighborhoods of $x$, for the conditional quantile of $Y$ given $X = x$. The performance of this approach therefore depends heavily on the density of the data around the considered point $x$ (\cite{gardes_estimating_2015}, \cite{daouia_kernel_2013-1}, \cite{gardes_integrated_2019}).  
The third combines quantile regression and extreme value theory, making certain assumptions about the behavior of the tails. These assumptions, which are often restrictive, such as linearity (\cite{chernozhukov_inference_2011}, \cite{wang_estimation_2012}) or tail equivalence (\cite{wang_estimation_2013}), may fail in real-world applications, particularly when the shape of the distribution varies according to the covariates. Finally, the fourth class is based on parametric models, under the assumption that the distribution tail follows either a generalized Pareto distribution (GPD) or a generalized extreme value distribution (GEV), with parameters depending on covariates, either parametrically (\cite{wang_estimation_2013}, \cite{wang_tail_2009}) or nonparametrically (\cite{gnecco_extremal_2024}, \cite{pasche_neural_2024}, \cite{velthoen_gradient_2023}, \cite{farkas_generalized_2024}). However, this approach requires delicate choices, such as the selection of thresholds or block sizes.   The conditional quantile of order $\tau$ is defined by $\mathcal{Q}_{\tau}(Y \mid X = x) = \inf \{ y \,: \, F_{Y \mid X = x}(y) \geq \tau \},$   and, when the conditional distribution $F_{Y \mid X = x}(\cdot)$ is continuous, this amounts to:
\begin{equation}\label{quant_invForm}
	\mathcal{Q}_{\tau}(Y|X=x)= F^{-1}_{Y|X=x}(\tau),
\end{equation}
for all $x \in \mathcal{X}$ and $\tau \in (0,1)$. For simplicity of notation, we will denote in the following $ \mathcal{Q}_{x}(\tau)$ to refer to the conditional quantile function $\mathcal{Q}_{\tau}(Y \mid X=x)$.   Extreme value theory facilitates extrapolation into the tail of the distribution and is widely used in the literature to address the second challenge of quantile regression mentioned above. For a detailed review, see for example \cite{beyerlein_quantile_2014-1} and \cite{chernozhukov_extremal_2017}.   Regarding the first challenge, several models based on statistical learning methods have been proposed (\cite{meinshausen_quantile_2006}, \cite{athey_generalized_2019}, \cite{chaudhuri_nonparametric_2002}). More recent works have attempted to address both challenges simultaneously, mainly by combining the Peak Over Threshold (POT) approach from extreme value theory with different statistical learning methods. Among these are models based on gradient boosting (\cite{velthoen_gradient_2023}), generalized additive models (\cite{youngman_generalized_2019}), neural networks (\cite{pasche_neural_2024}), and models using generalized random forests (\cite{gnecco_extremal_2024}).   The method proposed in this work is an extension of these approaches within the block maxima (BM) framework, which is preferred in many fields when modeling extremes, particularly in meteorology \cite{boudrissa_modelling_2017}, risk analysis \cite{calabrese_estimating_2015}, and hydrology.
More specifically, in order to facilitate extrapolation into the tail of the distribution and address the second challenge of quantile regression, we use the generalized extreme value (GEV) distribution, with parameters depending on the covariate $x$, to model block maxima. These parameters are then estimated via weighted likelihood, with weights obtained from the generalized random forest (grf) method introduced by \cite{athey_generalized_2019}.   Thus, the distribution $F_{Y \mid X=x}(\cdot)$ in equation~(\ref{quant_invForm}) is approximated by the GEV distribution, explicitly defined by:
\begin{equation*}\label{gev_condiDistribution}
	G_{(\xi(x),\mu(x),\sigma(x))}(z)=\begin{cases}
		\exp\left(-\left(1+\xi(x) \dfrac{z-\mu(x)}{\sigma(x)}\right)_{+}^{-\frac{1}{\xi(x)}}\right) &if \quad\xi(x) \neq0,\\
		\exp \left(-\exp \left(-\frac{z-\mu(x)}{\sigma(x)}\right)\right)& if \quad\xi(x)=0, \\
	\end{cases} 
\end{equation*}
defined on  $\Big\{z\in \mathbb{R}: 1+\xi(x) \dfrac{z-\mu(x)}{\sigma(x)}>0\Big\}$ for all   $ x \in \mathcal{X}$ with  $a_{+}=\max\{0,a\}.$  We estimate $\theta(x)=\left(\mu(x),\sigma(x),\xi(x)\right)$ by 
$$\hat{\theta}_n(x)\in \arg\max_{\theta \in \Theta} L_n(\theta;x)$$
where $L_n(\theta;x)$ is given by
$	L_n(\theta,x)=\sum_{i=1}^{n}w_i(x)\ell_{\theta}(z_i)$	
with  $\ell_{\theta}(z)=\log{\left(\dfrac{dG_{ \mu, \sigma, \xi}}{dz_i}(z)\right)},$ and $w_i(x)$ for all $i=1,\cdots,n,$ representing weights obtained through the grf method. These weights are used to address the first challenge, namely capturing the complex structure in the data and facilitating estimation when the dimension $p$ of the covariate space $\mathcal{X}$ is large. Based on this estimation, the conditional quantile is obtained by plugging the estimated GEV distribution with parameter $\hat{\theta}_n(x)$ into equation~(\ref{quant_invForm}), thereby yielding the conditional quantile estimator. 

The main contribution of this paper is therefore theoretical. Under a conditional max-domain-of-attraction assumption and suitable regularity conditions on the grf weights, the distribution of the covariates associated with the block maxima, and the block size, we establish the existence of a local weighted maximum likelihood estimator and prove its consistency. More specifically, for a fixed covariate value $x,$ we show that the estimated shape parameter converges in probability to the conditional extreme value index, while the location and scale parameters converge to the corresponding normalizing sequences. As a consequence, the conditional extreme quantile estimator obtained by inversion of the fitted GEV distribution is also consistent. The proof combines the asymptotic arguments developed for GEV maximum likelihood estimation by \cite{dombry_existence_2015} with properties of generalized random forest weights, particularly their localization and boundedness properties. This theoretical result provides a foundation for the use of grf-weighted GEV estimation in conditional extreme quantile regression. It also complements recent work on generalized random forest methods for extreme quantile regression (see \cite{vidagbandji_generalized_2025} and \cite{vidagbandji_penalized_2026-1}), which focused primarily on methodological and applied aspects. In contrast, the present paper focuses on the theoretical properties of the underlying unpenalized estimator.

The remainder of the paper is organized as follows. Section~\ref{section_background} introduces the notation and provides the necessary background on generalized random forests and block maxima approach from extreme value theory,  which are essential for establishing the consistency result of the proposed estimator, which will be presented in Section~\ref{section_mainContribution}. The Section~\ref{section_proof} is devoted to the proof of the existence and  the consistency theorem. Finally, Sections~\ref{section_simulation} and~\ref{section_realData} assess the performance of our approach through simulations and real data, comparing it with several existing methods.

\section{Proposed method}\label{section_background} 
We first introduce the statistical framework underlying the proposed estimator. We briefly recall generalized random forests for conditional quantile estimation and the block maxima approach from extreme value theory. These ingredients are then combined to construct a locally weighted GEV likelihood for conditional extreme quantile estimation.
\subsection{Quantile regression and Generalized random forest}\label{section_QRandGRF}
Let \((X_1,Y_1), \ldots, (X_n,Y_n)\)  be an $i.i.d.$ sample drawn form the joint distribution of  $(X,Y)$, where \(X_i \in\mathcal{X}   \subset  \mathbb{R}^p\) denotes a vector of covariates, and  \(Y_i \in \mathcal{Y} \subset \mathbb{R}\) the response variable. In quantile regression, the objectives is to estimate the conditional quantile of probability level $\tau_n \in (0,1),$ defined as \begin{equation}\label{quantile_infForm}
	\mathcal{Q}_{x}(\tau_n) = \inf \{ y : F_{Y|X=x}(y) \geq \tau_n \}
\end{equation} 
where  $F_{Y|X=x}(.)$ denotes the conditional distribution function of   $Y$ given $X = x.$
Another characterization of the conditional quantile is based on the following optimization problem: 
\[
\mathcal{Q}_{x}(\tau_n)  = \arg\min_q \, \mathbb{E}\left[\rho_{\tau_n}(Y, q) \mid X = x\right],
\]
where  $\rho_{\tau_n}(., .) $ is the asymmetric loss function defined as
\begin{equation}\label{quantile_loss}
	\rho_{\tau_n}(Y,q) = \begin{cases}
		\tau_n (Y-q),&   Y-q\geq 0 \\
		(1-\tau_n)(Y-q) ,&  Y-q<0.
	\end{cases}
\end{equation}
A variety of conditional quantile estimators have been proposed based on this formulation, under both parametric and nonparametric assumptions on the quantile function (\cite{chernozhukov_extremal_2005}, \cite{chernozhukov_extremal_2017}, \cite{buhai_quantile_2005}, \cite{wang_estimation_2012}, \cite{wang_estimation_2013}, \cite{fakoor_flexible_2023}).
Based on this characterization and  without parametric assumptions on the conditional distribution, a natural weighted approximation of $\mathcal{Q}_{x}(\tau)$ is
\[
\widehat{\mathcal{Q}}_{x}(\tau_n) = \arg\min_q \sum_{i=1}^n w_i(x) \rho_{\tau_n}(Y_i, q),
\]
where \(w_i(x)\) is a localized weight function reflecting the influence of the observation \(i\) in a neighborhood of \(x\).

Using the characterization in equation~(\ref{quant_invForm}), \cite{meinshausen_quantile_2006} proposed a method based on the random forests introduced by \cite{breiman_random_2001} to construct an empirical estimator of the conditional distribution function :
$\hat{F}(y \mid X = x) = \sum_{i=1}^n w_i(x) \mathbb{1}_{\{Y_i \leq y\}}.
$ This estimated distribution function is then plugged into equation~(\ref{quant_invForm}) to obtain an estimate of the conditional quantile. The weights $w_i(x)$ are obtained via the random forest method, which is based on conditional mean estimation ; hence $w_i(x)$ tends to be large for   $i \in \{1,\cdots,n\}$ whenever $\mathbb{E}\left(Y|X=x\right)~\approx~\mathbb{E}\left(Y|X=X_i\right)$ (\cite{lin_random_2006}, \cite{meinshausen_quantile_2006}). However, the situation differs when estimating conditional quantiles. As shown by  \cite{athey_generalized_2019} and \cite{gnecco_extremal_2024}, the weight $w_i(x)$ may not be large for
$i\in \{1,\cdots,n\},$ even if   
$\mathcal{Q}_{x}(\tau)~\approx ~\mathcal{Q}_{X_i}(\tau).$ Consequently, the classical random forest method of \cite{breiman_random_2001} fails to capture the complex structure of the quantile function. The generalized random forest, introduced by \cite{athey_generalized_2019} for quantile regression problems, overcomes this limitation and provides a more flexible and targeted estimation of heterogeneous effects, particularly for conditional quantiles.

The random forest method relies on aggregating the predictions of \( B \) 
decision trees to produce an estimator that is more stable, robust, and accurate than any single tree considered in isolation \cite{breiman_random_2001}. For an observation 
$ x \in \mathcal{X},$    the prediction $\eta_b(x)$ 
from the  $ b^{\text{th}} $ tree in the regression context is defined as the average of the responses  $Y_i$ corresponding to observations  $X_i$ that fall in the same terminal region $R_b(x)$ as $x:$ 
\[
\eta_b(x) = \sum_{i=1}^{n} \frac{ \mathbb{1}_{\{X_i \in R_b(x)\}} Y_i }{ \left| \{i : X_i \in R_b(x)\} \right| }, \quad b = 1, \ldots, B,
\]
where \( R_b(x) \subset \mathbb{R}^p \) denotes the leaf or terminal region containing  \( x \) in the tree \( b, \) and $|E|$ represents the cardinality of the set \( E \).
The final random forest estimator is obtained by averaging the predictions across all $B$ trees:
\[
\eta(x) = \sum_{i=1}^{n} w_i(x) Y_i,
\]
where the similarity weights \( w_i(x) \) are given by 
\begin{equation} \label{grf_weight}
	w_i(x) = \frac{1}{B} \sum_{b=1}^{B} \frac{ \mathbb{1}_{\{X_i \in R_b(x)\}} }{ \left| \{i : X_i \in R_b(x)\} \right| }.
\end{equation}
These weights  \( w_i(x) \) reflect the proximity between the observation 
\( x \) and  each observation \( X_i, \) as induced by the structure of the trees in the forest. The generalized random forest  approach extends the classical method by constructing trees based on a customizable loss function suited to the learning task. In the case of quantile estimation, the loss function is that given in equation~\eqref{quantile_loss}. The grf similarity weights  \( w_i(x) \) retain the same structure as in classical random forests, with the key difference being the splitting criterion used to partition the data. We employ this learning method to capture the complex structure of the quantile function and to facilitate estimation when the dimension of the covariate space is high. The specific properties of this method, which are essential for studying the consistency of the proposed estimator, will be recalled in Section~\ref{section_mainContribution}.

\subsection{Extreme quantile regression based on the BM approach}

The block maxima  method is a widely used approach for modeling extreme values across various fields. In this work, it will be employed to facilitate extrapolation in the tail of the distribution, thereby allowing the estimation of quantiles at probability levels $\tau_n \to 1$ such that $n(1-\tau_n)$ remains finite as $n \to +\infty$. This approach addresses the challenge of estimating conditional quantiles at extreme probability levels. For a more in-depth understanding of the extreme value theory (EVT) results presented in this section, we refer to \cite{de_haan_extreme_2006} and \cite{coles_introduction_2001-1}. We first recall some classical results in the unconditional case, before showing how this approach can be combined with the generalized random forest method to simultaneously tackle the issues identified in quantile regression.

The BM method is based on the limiting theorem describing the asymptotic behavior of the maximum of a sequence of independent and identically distributed random variables $Y_1, \dots, Y_m$ with an unknown probability distribution $F.$ This fundamental result originates from the pioneering work of \cite{fisher_limiting_1928} and \cite{gnedenko_sur_1943},  who established conditions such that there exist normalization sequences 
$a_m > 0$ and $b_m \in \mathbb{R}$ such that
\begin{equation} \label{Maxdomain_attracUnivar}
	\lim_{m \to +\infty} F^{m}\left(a_m y + b_m \right)= G_{\xi}(y), \quad \text{for all } y \in \mathbb{R},
\end{equation}
where \( G_{\xi} \) is a non-degenerate probability distribution defined by
\[
G_{\xi}(y) = \exp\left( - (1 + \xi y)^{-1/\xi} \right), \quad \text{ for  } 1 + \xi y > 0.
\] 
This law is known as the Generalized Extreme Value (GEV) distribution.  
A distribution function \(F\) that satisfies relation \eqref{Maxdomain_attracUnivar} is said to belong to the maximum domain of attraction of \(G_{\xi}\)  (\cite{de_haan_extreme_2006}), which is denoted by \(F \in \mathcal{D}(G_{\xi})\).  

To  illustrate the block maxima approach, consider a sequence of independent and identically distributed random variables \(Y_1, \dots, Y_N\), following a probability distribution \(F\) that belongs to the maximum domain of attraction \(\mathcal{D}(G_{\xi_0})\).  
The BM method consists in partitioning these data into \(n\) disjoint blocks of approximately equal size \(m > 1\), defined as:  
\[
B_{k,m} = \{Y_{(k-1)m+1}, \dots, Y_{km}\}, \quad k = 1, \dots, n.
\]  
For each block, we consider the maximum statistic:   $Z_k = \max_{i \in B_{k,m}} Y_i.$  
Since the \(Y_i\) are i.i.d., the distribution of the block maximum \(Z_k\) of size \(m\) is given by \(F^m\). According to the classical limit theorem for maxima, this distribution converges, after normalization, to the GEV distribution with parameters \((a_m, b_m, \xi_0)\).   The BM method assumes that the block maxima \(Z_1, \dots, Z_n\) follow a GEV distribution and are independent. The general form of the GEV distribution is given by:
\begin{equation} \label{gev_dist}
	G_{\mu, \sigma, \xi}(z) =
	\begin{cases}
		\exp\left( -\left(1 + \xi \dfrac{z - \mu}{\sigma} \right)_{+}^{-1/\xi} \right), & \xi \neq 0, \\
		\exp\left( -\exp\left( -\dfrac{z - \mu}{\sigma} \right) \right), & \xi = 0,
	\end{cases}
\end{equation} defined for $1 + \xi \dfrac{z - \mu}{\sigma}>0.$ The parameters $\mu \in  \mathbb{R},$ $\sigma>0$ and   $\xi \in \mathbb{R}$ are respectively the location, scale, and shapes parameters. The following section outlines how this approach is applied in the context of extreme quantile regression.

\section{Main  results}\label{section_mainContribution}
We combine the BM approach, which facilitates extrapolation in the tail of the distribution, with the generalized random forest method, which captures complex structures in the data and improves estimation when the covariate space is high-dimensional. To illustrate the proposed method, consider a sequence of i.i.d. random vectors \((X_1, Y_1), \dots, (X_N, Y_N)\) distributed as the random vector $(X, Y)$, where $X \in \mathcal{X} \subset \mathbb{R}^p$ and $Y \in \mathcal{Y} \subset \mathbb{R}$. We assume that the sample size is $N = n \times m$, and that we form $n$ blocks of size $m$, so that the $k^{\text{th}}$ block, for $k = 1, \dots, n$, is given by
\begin{equation*}
	B_{k,m}=\{(X_{(k-1)m+1},Y_{(k-1)m+1}),\cdots,(X_{km},Y_{km})\}
\end{equation*}	
Let us define
\begin{equation}\label{max_Z_km}
	Z_{k,m}=\max \{ Y_i: (X_i,Y_i)\in B_{k,m}\}	
\end{equation}
and  $X_{k,m}$ denote the covariate $X_i$ corresponding to the maximizing $Y_i$ in $\{ Y_i : (X_i, Y_i) \in B_{k,m} \}$.  
Thus, $(X_{k,m}, Z_{k,m})$ for $k = 1, \dots, n$ forms the block maxima sample with respect to the response variable $Y.$  For all $x \in \mathcal{X}$, let $F_x(\cdot)$ denote the conditional distribution of $Y$ given $X = x.$ 
\begin{assump}[Conditional max-domain of attraction]\label{assump1}
	For every $x \in \mathcal{X}$, the conditional distribution function 
	$F_x(.)$ 
	belongs to the max-domain of attraction of an extreme value distribution 
	$G_{\xi_0(x)}.$
	Equivalently, there exist normalizing functions $a_m(x) > 0$ and 
	$b_m(x) \in \mathbb{R}$ such that, for all  $z\in\mathbb{R}$,
	\begin{equation}\label{max_domConv}
		\lim_{m \to +\infty}
		F_x^m\left(a_m(x)z + b_m(x)\right)
		=
		G_{\xi_0(x)}(z),
	\end{equation}
	where
	\[
	G_{\xi_0(x)}(z)
	=
	\exp\left(
	- \big(1 + \xi_0(x) z\big)^{-1/\xi_0(x)}
	\right),
	\qquad
	1 +  z\xi_0(x) > 0.
	\]
	In the case $\xi_0(x)=0$, the distribution $G_{\xi_0(x)}$ is understood as 
	the continuous limit of $G_{\xi_0}(z)$ as $\xi_0 \to 0.$
\end{assump}
This assumption is the conditional version of the convergence in (\ref{Maxdomain_attracUnivar}) and implies that the normalized block maxima
converge in distribution to $G_{\xi_0(x)}$ as $m \to +\infty$.
\begin{rem}\label{remarqDefi_Fm} Suppose $F_x\in \mathcal{D}(G_{\xi_0(x)}) $ and  denote the  corresponding normalization sequences $a_m(x)$ and $b_m(x)$. 
	The   block maximum  $Z_{k,m}$ given $X_{k,m}=x$ has distribution function $F_x^m(z)$ which we shall denote by $F_m(z|x)$ for notational simplicity. In other words, for any $m\geq 1,$ we have
	\[
	F_m(z \mid x)
	:=
	F_x^m\big(z\big).
	\]
	Note that $F_m(\cdot \mid X_{k,m})$ is the true conditional distribution function of $Z_{k,m}$ given $X_{k,m}.$ 	Under Assumption~\ref{assump1}, we have the  convergence
	\begin{equation}\label{limit_of_Fm}
		\lim_{m \to +\infty}
		F_m\left(a_m(x) z + b_m(x) \mid x\right)
		=
		G_{\xi_0(x)}(z),
		\qquad z \in \mathbb{R}.
	\end{equation}
\end{rem} 

We model the conditional distribution of $Z_{k,m}$ given $X_{k,m}=x$ by a generalized  extreme values  distribution with parameters $(\mu(x),\sigma(x),\xi(x)),$ when $m\to \infty$.
In the conditional setting, the parameters $\mu$, $\sigma$, and $\xi$ of the GEV distribution depend on $x$ and are described by the functions  $\mu : \mathcal{X} \to \mathbb{R}$, $\sigma : \mathcal{X} \to \mathbb{R}^*_+$, and $\xi : \mathcal{X} \to \mathbb{R}$.   To address the extrapolation problem in the tail, we approximate the conditional distribution in (\ref{quant_invForm}) by this conditional GEV distribution.   We then define the estimator $\hat{\theta}_n(x)=\left(\hat{\mu}_n(x), \hat{\sigma}_n(x), \hat{\xi}_n(x)\right)$ 
of the conditional GEV parameters as the value of $\theta$ that maximizes the local weighted log-likelihood function defined by
\begin{equation}\label{log_vraisemblancePond}
	L_n(\theta;x)=\sum_{k=1}^{n}w_k(x)\ell_{\theta}(z_{k,m}), \quad x \in \mathcal{X},	
\end{equation}
where  $\ell_{\theta}(z)$ is equal to : 
$$\ell_{\theta}(z)=\begin{cases}
	-\log(\sigma)-(1+\frac{1}{\xi})\log \left(1+\xi\frac{z-\mu}{\sigma}\right)-\left(1+\xi\frac{z-\mu}{\sigma}\right)^{-\frac{1}{\xi}}&  \text{if}\quad  1+\xi\frac{z-\mu}{\sigma}>0,\\
	-\infty & \text{otherwise,}
\end{cases}$$ when  $\xi \neq 0,$ and
$$\ell_{\theta}(z)= -\log(\sigma)-\left(\frac{z-\mu}{\sigma}\right)-\exp{\left(-\frac{z-\mu}{\sigma}\right)}$$ when $\xi= 0.$ 
The weights $w_k(x)$ are obtained using the generalized random forest method  briefly described in Section~\ref{section_QRandGRF}. The weights $w_k(x)$ depend on the sample size~$n$, since they are derived from the generalized random forest built from the entire dataset.    
The conditional quantile of probability level $\tau_n$, with $\tau_n$ close to $1$, is then obtained by substituting the estimated GEV distribution into equation~(\ref{quant_invForm}) and is given by \begin{equation*}
	\widehat{\mathcal{Q}}^Z_{x}(\tau_n)= \begin{cases}
		\hat{\mu}(x) + \dfrac{\hat{\sigma} (x)}{\hat{\xi} (x)}\left(\left(-\log(\tau_n)\right)^{-\hat{\xi} (x)}-1\right)  & \text{if}\quad  \hat{\xi} (x)\neq 0,\\
		\hat{\mu}(x) - \hat{\sigma}(x)\log \left( -\log\left(\tau_n\right)\right)  & \text{if}\quad \hat{\xi}(x)= 0. 
\end{cases}\end{equation*} The conditional quantiles of the original data are therefore given, for all $x \in \mathcal{X},$ by
\begin{equation}\label{GEVestiquantile}
	\widehat{\mathcal{Q}}^Y_{x}(\tau_n)= \begin{cases}
		\hat{\mu}(x) + \dfrac{\hat{\sigma} (x)}{\hat{\xi} (x)}\left(\left(-\log(\tau_n^m)\right)^{-\hat{\xi} (x)}-1\right)  & \text{if}\quad  \hat{\xi} (x)\neq 0,\\
		\hat{\mu}(x) - \hat{\sigma}(x)\log \left( -\log\left(\tau_n^m\right)\right)  & \text{if}\quad \hat{\xi}(x)= 0. 
\end{cases}\end{equation}

The support of the GEV distribution is given by 
$\Big\{ z \in \mathbb{R} : 1 + \xi \tfrac{z - \mu}{\sigma} > 0 \Big\}, \quad \text{if } \xi \neq 0,$
and by $\mathbb{R}$ if $\xi = 0$. This support depends on the parameters of the GEV distribution, which are unknown in practice. This makes it difficult to establish the usual regularity conditions required to guarantee the classical asymptotic properties of the maximum likelihood estimator (MLE).  It should be noted that the weighted likelihood defined in equation~(\ref{log_vraisemblancePond}) does not admit a global solution. Indeed, \cite{smith_maximum_1985} showed that the MLE does not exist if $\xi \leq -1$, that it is asymptotically normal when $\xi > -0.5$, and that it nevertheless remains consistent for $\xi > -1$. Along the same lines, \cite{dombry_existence_2015} proved the existence and consistency of the MLE of the GEV parameters for all $\xi > -1$, under a first-order extreme value condition. Recently, \cite{dombry_maximum_2019} provided a formal proof of the asymptotic normality of the GEV MLE.  We say that $\hat{\theta}(x)$ is a local weighted maximum likelihood estimator if it maximizes (\ref{log_vraisemblancePond}) over $\Theta$, where  $\Theta \subset \mathbb{R} \times (0,+\infty) \times (-1,+\infty).$ This local estimator is defined by
\begin{equation}\label{MLE}
	\hat{\theta}_n(x)\in \arg\max_{\theta \in \Theta} L_n(\theta;x), \quad \forall x\in \mathcal{X}.
\end{equation}
In order to establish the consistency of this estimator, we introduce a set of assumptions on the weights $w_k(\cdot)$ that ensure the convergence of $\hat{\theta}_n(x)$. 
Assumptions~\ref{assump_diam_convProba} and~\ref{assump_wk} gather the regularity conditions on the forest weights required in our proof. 
These assumptions are adapted from \cite{gnecco_extremal_2024} and \cite{athey_generalized_2019}.
\begin{assump} \label{assump_diam_convProba}
	Let $b = 1, \dots, B$ denote a tree in the forest, and let $x \in \mathcal{X}$ be a fixed predictor point. 	Define the diameter of the leaf $R_b(x)$ by
	\[
	\mathrm{diam}(R_b(x)) := \sup \{ \|y - x\|_2 : y \in R_b(x) \}.
	\]
	Let $s < n$ denote the number of observations used to grow the tree. 
	We assume that the diameter of the leaf $R_b(x)$ converges in probability to zero, that is, for every $\varepsilon > 0$,
	$$
	\mathbb{P}[\mathrm{diam}(R_b(x)) > \varepsilon] \to 0 \quad \text{as } s \to \infty.
	$$	
\end{assump}

\begin{assump}\label{assump_wk} \label{Honesty}
	For every $x\in\mathcal{X}$,	the weights satisfy $\sum_{k=1}^{n} w_k(x) = 1$ and, for all $k = 1, \dots, n$, we have
	\[
	\max_{1 \leq k \leq n} w_k(x) \leq \frac{s}{n}.
	\]	
	where $s$ denotes the subsample size used in the construction of the trees, with $s \to \infty$ and $ s/n \to 0$ as $n \to \infty$. 
	We assume that the forest is honest, following the framework of \cite{wager_estimation_2018} and \cite{athey_generalized_2019}.  Thus, for each $k$, the weight $w_k(x)$ is independent of $Z_{k,m}$ conditional on $X_{k,m}=x$. 
\end{assump} 	
Finally, we introduce assumptions regarding the distribution of the covariate $X_{k,m}$ concomitant to the maximum, as well as the conditional distribution $F_m(\cdot \mid X_{k,m} = x)$. These conditions are essential for the proof of Lemma~\ref{lemme2} and Lemma~\ref{lemma4} and are analogous to those considered by \cite{gnecco_extremal_2024} and \cite{meinshausen_quantile_2006}. Note that $X_{k,m}$ is the covariate associated with the block maximum, which implies that its distribution might differ from the marginal distribution of $X.$ We assume in the assumption below that it still admits a well-behaved density.

\begin{assump}
	\label{assump_cond_dist_Fm}
	Assume that the predictor space $\mathcal X\subset\mathbb R^p$ is compact and that the distribution of the covariate $X_{k,m}$ associated with the block maximum admits a density on $\mathcal X$ that is uniformly bounded away from zero and infinity in $m$.
	Moreover, assume that there exists a constant $L>0$ such that $F_m(y\mid X=x)$ is Lipschitz continuous in $x$, with Lipschitz constant $L$, uniformly for all sufficiently large $m$. That is, for all $x,x'\in\mathcal X$,
	$$
	\sup_{y\in\mathbb R}
	\left|
	F_m(y\mid X=x)-F_m(y\mid X=x')
	\right|
	\leq L\|x-x'\|_2.
	$$
	Assume that, for every $x\in\mathcal X$, the conditional distribution function $F_m(y\mid X=x)$ is strictly increasing in $y$.
	Assume that the family of functions \(\{f_m\}_{m \geq 1}\) defined by
	$$
	f_m(x) = \mathbb{E}\!\left[ Y_{k,m} \mid X_{k,m} = x \right], 
	\qquad x \in \mathcal{X},
	$$
	is equicontinuous on \(\mathcal{X}\), where
	$$
	Y_{k,m}
	=
	\ell_{\theta_0(x)}
	\left(
	\frac{Z_{k,m}-b_m(x)}{a_m(x)}
	\right),
	\qquad
	\theta_0(x)=(0,1,\xi_0(x)).
	$$	
	Finally, assume that there exists a constant $C>0$ such that, for all $m$ and all $k=1,\ldots,n$,	
	\[
	\mathbb{E}\Bigl[\bigl|Y_{k,m}\bigr|^2 \;\Big|\;X_{k,m}\Bigr]\le C. \tag{H}\label{Hypothese_H}
	\]
	
\end{assump}

In the BM approach, the block size $m$ must be chosen carefully: if \( m \) is too large, the variance of the estimator increases ; if it is too small, the bias becomes more significant. Therefore, a bias–variance trade-off must be considered when selecting $ m.$ \cite{dombry_existence_2015} proved consistency under the condition \(\tfrac{m}{\log(n)} \to +\infty\), or equivalently \(\tfrac{n \log(n)}{N} \to 0\). The presence of the logarithmic factor implies that, in practice, the number of observations effectively used in the block maxima approach is of a smaller order than that required by the POT method, which only requires \(\tfrac{n}{N} \to 0\), where \( n \) denotes in this setting the number of threshold exceedances (\cite{zhou_existence_2009}, \cite{gnecco_extremal_2024}).   This condition is crucial for establishing the consistency of the estimator proposed in this work, which is stated below. Furthermore, we assume that the block size $m$ depends on the sample size $n$ of the block maxima, that is, $ m = m(n).$	
\begin{theo}[Existence and Consistency of $\hat{\theta}_n(x)$] \label{theo_consistency} 
	Let \( x \in \mathcal{X} \). Suppose $
	F_x \in \mathcal{D}(G_{\xi_0(x)})$  and 
	assume that Assumptions \ref{assump1} to \ref{assump_cond_dist_Fm} hold with \(\xi_0(x) > -1\) and  
	$\lim\limits_{n \to +\infty} \frac{m(n)}{\log(n)} = +\infty.$  
	Then, there exists a sequence of estimators \((\hat{\mu}_n(x), \hat{\sigma}_n(x), \hat{\xi}_n(x))\) and an integer \( N \geq 1 \) such that  
	\begin{equation}\label{equation_existence}
		\mathbb{P}\left[ (\hat{\mu}_n(x), \hat{\sigma}_n(x), \hat{\xi}_n(x)) \text{ is a  local MLE for all } n \geq N \right] = 1,
	\end{equation}
	and  
	\begin{equation}\label{equation_consistence}
		\hat{\xi}_n(x) \inprob \xi_0(x), 
		\qquad \frac{\hat{\mu}_n(x) - b_m(x)}{a_m(x)} \inprob 0, 
		\qquad \frac{\hat{\sigma}_n(x)}{a_m(x)} \inprob 1, 
		\quad \text{as } n \to +\infty.
	\end{equation}		
\end{theo}
The condition  $\lim\limits_{n \to +\infty} \frac{m(n)}{\log(n)} = +\infty$  
comes from the work of \cite{dombry_existence_2015} and is required in the proof of Lemma~\ref{lemma4}.  The assumption \(\xi_0(x) > -1\) is crucial and aligns with the results of \cite{dombry_existence_2015} in the unconditional case, as it ensures the first-order condition. That is, the maximum likelihood estimator must satisfy  
\[
\left( 
\frac{\partial Ln}{\partial \mu}(\theta; x), \, 
\frac{\partial Ln}{\partial \sigma}(\theta; x), \, 
\frac{\partial Ln}{\partial \xi}(\theta; x) 
\right) = (0,0,0).
\]  
Indeed, we have
\begin{eqnarray}
	\notag
	\frac{\partial L_n}{\partial \mu}(\theta; x) 
	&=& \sum_{k=1}^{n} w_k(x)\frac{\partial \ell_{\theta}}{\partial \mu}(Z_{k,m}) \\ \notag
	&=& -\dfrac{1}{\sigma} \sum_{k=1}^{n} w_k(x) \left[ 
	-\left( \dfrac{1+\xi}{1+\xi \frac{Z_{k,m}-\mu}{\sigma}} \right) 
	+ \left(1+\xi \frac{Z_{k,m}-\mu}{\sigma}\right)^{-1-\frac{1}{\xi}} 
	\right] \\ \notag
	&=& -\dfrac{1}{\sigma} \sum_{k=1}^{n} w_k(x) \,
	g'_{\xi}\!\left(\dfrac{Z_{k,m}-\mu}{\sigma}\right),
\end{eqnarray}  
with  
\[
g_{\xi}(y) = -\left(1+\frac{1}{\xi}\right)\log\!\left(1+\xi y\right) 
- \left(1+\xi y\right)^{-\frac{1}{\xi}}.
\]  
According to Proposition~1 of \cite{dombry_existence_2015}, when \(\xi \leq -1\), the function \(g_{\xi}\) is strictly increasing on its domain of definition. Consequently, for any \(x \in \mathcal{X}\), we have  
\[
\frac{\partial Ln}{\partial \mu}(\theta, x) < 0 \quad \text{whenever } \xi \leq -1.
\]  
It follows that the condition \(\xi > -1\) is necessary for the existence of a local extremum.
\begin{rem}
	Although empirical evidence in Section~\ref{expe1} supports the asymptotic normality of the weighted maximum likelihood estimator, a rigorous theoretical proof demands a deeper asymptotic investigation. Adapting the framework of \cite{dombry_maximum_2019} to our forest-weighted setup is left for future research.
\end{rem}

\begin{cor}	\label{cor:quantile_consistency}	
	Let $x\in\mathcal X$ be fixed, and let $\tau\in(0,1)$ be fixed and close to $1$. The estimated conditional GEV quantile at probability level $\tau$ is defined by
	\begin{equation*}
		\widehat{\mathcal{Q}}^Z_{x}(\tau)= \begin{cases}
			\hat{\mu}(x) + \dfrac{\hat{\sigma} (x)}{\hat{\xi} (x)}\left(\left(-\log(\tau)\right)^{-\hat{\xi} (x)}-1\right)  & \text{if}\quad  \hat{\xi} (x)\neq 0,\\
			\hat{\mu}(x) - \hat{\sigma}(x)\log \left( -\log\left(\tau\right)\right)  & \text{if}\quad \hat{\xi}(x)= 0. 
	\end{cases}\end{equation*}	
	Under the assumptions of Theorem~\ref{theo_consistency}, it holds that
	$$
	\dfrac{\widehat{\mathcal{Q}}^Z_{x}(\tau)-\mathcal{Q}^Z_{x}(\tau)}{a_m(x)}
	\inprob
	0,
	$$
	where $\mathcal{Q}^Z_{x}(\tau)$ denotes the $\tau$-quantile associated to the true conditional GEV distribution with parameters $\bigl(a_m(x),b_m(x),\xi_0(x)\bigr)$.
	
\end{cor}

\begin{proof}
	Let $x \in \mathcal{X}$ and $\tau \in (0,1)$ be fixed. Define the function $g : (-1, +\infty) \to \mathbb{R}$ by:
	\begin{equation*}
		g(\xi) = \begin{cases}
			\dfrac{\left(-\log(\tau)\right)^{-\xi} - 1}{\xi} & \text{if } \xi \neq 0, \\[10pt]
			-\log\left(-\log(\tau)\right) & \text{if } \xi = 0.
		\end{cases}
	\end{equation*}	
	By a Taylor expansion of $\xi \mapsto \left(-\log(\tau)\right)^{-\xi} = \exp\left(-\xi \log\left(-\log(\tau)\right)\right)$ around $0$, we have $\lim_{\xi \to 0} g(\xi) = g(0)$. Thus, $g$ is continuous on $(-1, +\infty)$.
	Using this notation, the true conditional GEV quantile $\mathcal{Q}^Z_x(\tau)$ associated with parameters $\bigl(a_m(x), b_m(x), \xi_0(x)\bigr)$ and the estimated conditional quantile $\widehat{\mathcal{Q}}^Z_x(\tau)$ can be expressed respectively as:
	\begin{align*}
		\mathcal{Q}^Z_x(\tau) &= b_m(x) + a_m(x) g(\xi_0(x)), \\[5pt]
		\widehat{\mathcal{Q}}^Z_x(\tau) &= \hat{\mu}_n(x) + \hat{\sigma}_n(x) g(\hat{\xi}_n(x)).
	\end{align*}	
	Dividing $\widehat{\mathcal{Q}}^Z_x(\tau) -\mathcal{Q}^Z_x(\tau)$ by the normalizing sequence $a_m(x) > 0$ yields
	\begin{align*}
		\frac{\widehat{\mathcal{Q}}^Z_x(\tau)-\mathcal{Q}^Z_x(\tau)}{a_m(x)} &= \frac{\hat{\mu}_n(x) - b_m(x)}{a_m(x)} + \left(\frac{\hat{\sigma}_n(x)}{a_m(x)}\right) g\left( \hat{\xi}_n(x)\right)- g\left(\xi_0(x)\right). 
	\end{align*}	
	Under the assumptions of Theorem~\ref{theo_consistency}, as $n \to +\infty$ (and $m(n) \to +\infty$), we have
	\begin{enumerate}
		\item[(i)] $\dfrac{\hat{\mu}_n(x) - b_m(x)}{a_m(x)} \inprob 0$;
		\item[(ii)] $\dfrac{\hat{\sigma}_n(x)}{a_m(x)} \inprob 1$;
		\item[(iii)] $\hat{\xi}_n(x) \inprob \xi_0(x)$, which implies $g(\hat{\xi}_n(x)) \inprob g(\xi_0(x))$ by the Continuous Mapping Theorem.
	\end{enumerate}
	Consequently, we obtain
	\begin{equation*}
		\frac{\widehat{\mathcal{Q}}^Z_x(\tau)-\mathcal{Q}^Z_x(\tau)}{a_m(x)} \inprob 0.
	\end{equation*}
	This completes the proof of Corollary \ref{cor:quantile_consistency}.
	
\end{proof}

\section{Proofs}\label{section_proof}
The proof of Theorem~\ref{theo_consistency} builds on the work of \cite{dombry_existence_2015}. 
We first introduce some notation and recall a few preliminary results that will be used in the proof.
\begin{nota}
	Let $x \in \mathcal{X}$. We normalize the block maxima as follows 
	$$ \tilde{Z}_{k,m} = \frac{Z_{k,m} - b_m(x)}{a_m(x)}, \quad \text{where $Z_{k,m}$ is defined in equation (\ref{max_Z_km})}.$$ 
	The associated weighted log-likelihood is given by 
	\begin{equation}\label{log_vraisemblance_normalise}
		\tilde{L}_n(\theta;x) = \sum_{k=1}^{n} w_k(x)\,\ell_{\theta}(\tilde{Z}_{k,m}),	
	\end{equation}
	where $w_k(x),$ for all $x \in \mathcal{X},$ denotes the weight assigned to each \( X_k \), with \( k = 1, \dots, n \), obtained using the GRF method, whose explicit form is provided in relation~(\ref{grf_weight}). 
	We therefore introduce
	$\varTheta = \mathbb{R} \times (0, +\infty) \times (-1, +\infty),$
	where a generic element is denoted by $\theta = (\mu, \sigma, \xi).$   The restriction of \(\tilde{L}_n(\cdot;x)\) to \(\varTheta\), that is, \(\tilde{L}_n(\cdot;x): \varTheta \to [-\infty, +\infty)\), is a continuous function. Hence, for any compact set \(K \subset \varTheta\), this restriction is bounded and reaches its maximum on \(K\).   Denoting this maximum by 
	$\tilde{\theta}_n^K(x) = \left(\tilde{\mu}_n^K(x), \tilde{\sigma}_n^K(x), \tilde{\xi}_n^K(x)\right),
	$ it is defined for all $x \in \mathcal{X}$ as follows
	\begin{equation}\label{theta_K_n}
		\tilde{\theta}_n^K(x)=\arg\max_{\theta\in K}\tilde{L}_n(\theta;x).
	\end{equation}
	For $x \in \mathcal{X}$, let us denote by $\mathbb{P}_n^x$ the empirical distribution defined as
	\[
	\mathbb{P}_n^x = \sum_{k=1}^{n} w_k(x) \, \delta_{\tilde{Z}_{k,m}},
	\]
	where $\delta_x$ represents the Dirac measure at the point $x \in \mathbb{R}$. The measure $P_n^x$ is a random probability measure conditional on $(X_1, \ldots, X_n)$, 
	since the weights satisfy $\sum_{k} w_k(x) = 1$.
	For any measurable function $f: \mathbb{R} \to [-\infty, +\infty)$, we define
	\[
	\mathbb{P}_n^x[f] = \sum_{k=1}^{n} w_k(x) \, f(\tilde{Z}_{k,m}).
	\]
	With these notations, the weighted log-likelihood defined in (\ref{log_vraisemblance_normalise}) can be written as
	\[
	\tilde{L}_n(\theta;x) = \mathbb{P}_n^x[\ell_{\theta}],
	\]
	and the associated empirical distribution function is given by
	\[
	\mathbb{F}_n^x(t) = \mathbb{P}_n^x((-\infty, t]) = \sum_{k=1}^{n} w_k(x) \, \mathbb{1}_{\{\tilde{Z}_{k,m} \leq t\}}, \quad \text{for all } t \in \mathbb{R}.
	\]
	We now present some essential results that will be used to establish the proof of Theorem~\ref{theo_consistency}.
\end{nota}

\begin{lem}\label{lemma1} 
	Let $x \in \mathcal{X}$. The triplet $(\hat{\mu}_n(x),\hat{\sigma}_n(x),\hat{\xi}_n(x))$ is a weighted maximum likelihood estimator  if and only if $\tilde{L}_n$ has a local maximum at
	$	\left(\frac{\hat{\mu}_n(x)-b_m(x)}{a_m(x)}, \frac{\hat{\sigma}_n(x)}{a_m}, \hat{\xi}_n(x)\right).$
\end{lem}

\begin{proof}
	Let $x \in \mathcal{X}$. According to Lemma~1 in \cite{dombry_existence_2015}, we have
	\[
	\ell_{(\mu, \sigma, \xi)}\Big(\frac{z-b}{a}\Big) = \ell_{(a\mu+b, a\sigma, \xi)}(z) + \log(a), \quad \text{for all } z \in \mathbb{R}.
	\]
	Hence,
	\begin{align*}
		\tilde{L}_n\Big(\Big(\frac{\mu-b_m}{a_m}, \frac{\sigma}{a_m}, \xi\Big); x\Big) 
		&= \sum_{k=1}^{n} w_k(x) \, \ell_{(\frac{\mu-b_m}{a_m}, \frac{\sigma}{a_m}, \xi)}\Big(\frac{Z_{k,m}-b_m}{a_m}\Big) \\
		&= \sum_{k=1}^{n} w_k(x) \, \ell_{(\mu, \sigma, \xi)}(Z_{k,m}) + \log(a_m), \quad \text{since } \sum_{k=1}^{n} w_k(x) = 1 \\
		&= L_n(\theta; x) + \log(a_m).
	\end{align*}
	It follows that
	\[
	L_n((\mu, \sigma, \xi); x) = \tilde{L}_n\Big(\Big(\frac{\mu-b_m}{a_m}, \frac{\sigma}{a_m}, \xi\Big); x\Big) - \log(a_m).
	\]
	Therefore, the local maxima of $L_n$ and $\tilde{L}_n$ correspond directly, which proves Lemma~\ref{lemma1}.
\end{proof}

\begin{lem}\label{lemme2} 
	Let $x \in \mathcal{X}$. Suppose that Assumptions~\ref{assump1} to~\ref{assump_cond_dist_Fm} hold and $\lim\limits_{n\to +\infty} m(n) = +\infty$. Then,
	\[
	\sup_{t \in \mathbb{R}}|\mathbb{F}_n^x(t)- G_{\xi_0(x)}(t)| \inprob 0. 
	\]
\end{lem}

\begin{proof}
	Let $x\in\mathcal{X}.$ Let $F_x \in \mathcal{D}(G_{\xi_0(x)})$, and let $a_m(x)$ and $b_m(x)$ be the corresponding normalizing sequences. 
	We write $F_m(\cdot \mid x)$ for the conditional distribution function of the maximum $Z_{k,m}$ given $X_{k,m} = x$, as introduced in Remark~\ref{remarqDefi_Fm}.
	For all $t \in \mathbb{R},$ we have 
	\begin{eqnarray}\notag
		\left|\mathbb{F}_n^x(t) - G_{\xi_0(x)}(t)\right|
		&=& \left| \sum_{k=1}^{n} w_k(x)\mathbb{1}_{\{\tilde{Z}_{k,m} \leq t\}} - F_m\left(a_m(x)t+b_m(x)\mid x\right) + F_m\left(a_m(x)t+b_m(x)\mid x\right) - G_{\xi_0(x)}(t) \right| \\\notag
		&\leq& |A_{n,m}(t)| + |B_{n,m}(t)|,
	\end{eqnarray}
	with 
	$$A_{n,m}(t)=\sum_{k=1}^{n} w_k(x)\mathbb{1}_{\{\tilde{Z}_{k,m} \leq t\}} - F_m\left(a_m(x)t+b_m(x)\mid x\right) \quad \text{and} \quad B_{n,m}(t)=F_m\left(a_m(x)t+b_m(x)\mid x\right) - G_{\xi_0(x)}(t).$$	
	Under Assumption~\ref{assump1}, and in particular from equation~\eqref{limit_of_Fm}, we obtain that for every $x \in \mathcal{X}$ and 
	$m = m(n) \to \infty$,
	\[
	\sup_{t \in \mathbb{R}} |B_{n,m}(t)| \longrightarrow 0.
	\]
	It therefore remains to prove that 
	$
	\sup_{t \in \mathbb{R}} |A_{n,m}(t)| \inprob 0.
	$ The proof of this convergence follows the same strategy as that of 
	Theorem~1 in \cite{meinshausen_quantile_2006}, which establishes the 
	uniform convergence in probability of a weighted conditional empirical 
	distribution function to the true conditional distribution in a 
	non-extreme setting, where the weights come from the classic random forest method. To this end, define the random variables $U_{k,m}$, $k=1,\dots,n$, as the 
	conditional probability integral transforms of the observations 
	$Z_{k,m}$ (the non‑normalized block maxima) given $X_{k,m}$:
	\[
	U_{k,m} = F_m(Z_{k,m} \mid X_{k,m}),
	\]
	where $F_m(\cdot \mid X_{k,m})$ denotes the true conditional distribution function of $Z_{k,m}$ given $X_{k,m}$. 
	By construction, conditionally on $X_{k,m}$, the variables $U_{k,m}$ are i.i.d. uniformly distributed on $[0,1]$, and they are independent across $k$.
	Now, for a fixed $t \in \mathbb{R}$, set $y = a_m(x)t + b_m(x)$. Then, we have 
	\[
	\{\tilde{Z}_{k,m} \le t\} = \left\{Z_{k,m} \le a_m(x)t + b_m(x)\right\} = \left\{U_{k,m} \le F_m(y \mid X_{k,m})\right\}.
	\]
	Therefore, the weighted empirical distribution function can be written as
	\begin{eqnarray}\notag
		\mathbb{F}_n^x(t) &=& \sum_{k=1}^n w_k(x) \mathbb{1}_{\{\tilde{Z}_{k,m} \le t\}} \\\notag
		&=& \sum_{k=1}^n w_k(x) \mathbb{1}_{\left\{U_{k,m} \le F_m(y \mid X_{k,m})\right\}}.
	\end{eqnarray}	
	Hence,
	\begin{align*}
		|A_{n,m}(t)|
		&=
		\Bigg|
		\sum_{k=1}^n w_k(x)
		\Big(
		\mathbb{1}_{\left\{U_{k,m} \le F_m(y \mid X_{k,m})\right\}}
		-
		\mathbb{1}_{\left\{U_{k,m} \le F_m(y \mid x)\right\}}
		\Big) \\
		&\quad
		+
		\sum_{k=1}^n w_k(x)
		\mathbb{1}_{\left\{U_{k,m} \le F_m(y \mid x)\right\}}
		-
		F_m(y \mid x)
		\Bigg| \\
		&\le
		\left|
		\sum_{k=1}^n w_k(x)
		\left(
		\mathbb{1}_{\left\{U_{k,m} \le F_m(y \mid X_{k,m})\right\}}
		-
		\mathbb{1}_{\left\{U_{k,m} \le F_m(y \mid x)\right\}}
		\right)
		\right|
		\\&\quad 
		+
		\left|
		F_m(y \mid x)
		-
		\sum_{k=1}^n w_k(x)
		\mathbb{1}_{\left\{U_{k,m} \le F_m(y \mid x)\right\}}
		\right|.
	\end{align*}	
	Taking the supremum over $t \in \mathbb{R}$ in the second term of the above inequality, we obtain
	\[
	\sup_{y \in \mathbb{R}}
	\left|
	F_m(y \mid x)
	-
	\sum_{k=1}^n w_k(x)
	\mathbb{1}_{\left\{U_{k,m} \le F_m(y \mid x)\right\}}
	\right|
	=
	\sup_{z \in [0,1]}
	\left|
	z
	-
	\sum_{k=1}^n w_k(x)
	\mathbb{1}_{\left\{U_{k,m} \le z\right\}}
	\right|,
	\]	
	since $F_m(\cdot \mid x)$ is strictly increasing (Assumption~\ref{assump_cond_dist_Fm}) and maps $\mathbb{R}$ onto $(0,1)$. 
	This term corresponds to the variance-type term studied in \cite{meinshausen_quantile_2006}. 
	The authors showed, after applying Bonferroni's inequality, that this term $\sup_{z \in [0,1]}
	\left|
	z- \sum_{k=1}^n w_k(x) \mathbb{1}_{\{U_{k,m} \le z\}} \right|$ converges to zero in probability provided that,
	\[
	\sum_{k=1}^n w_k^2(x) \longrightarrow 0 
	\qquad \text{as } n \to \infty.
	\]
	This condition is satisfied by the generalized random forest weights $w_k(x)$ used in our setting. 
	Indeed, by Assumption~\ref{assump_wk}, we have
	\[
	\max_{1 \le k \le n} w_k(x) \le \frac{s}{n},
	\qquad \text{with } \frac{s}{n} \to 0 \text{ as } n \to \infty.
	\]
	Since $\sum_{k=1}^n w_k(x)=1$, it follows that
	\begin{equation}\label{resul_wk2_to_0}
		\sum_{k=1}^n w_k^2(x)
		\le 
		\max_{1 \le k \le n} w_k(x)
		\longrightarrow 0,
		\qquad n \to \infty.
	\end{equation}
	Consequently, for every $x \in \mathcal{X}$,
	\[
	\sup_{z \in [0,1]}
	\left|
	z - \sum_{k=1}^n w_k(x)\mathbb{1}_{\{U_{k,m} \le z\}}
	\right|
	\inprob 0.
	\]	
	We now turn to the first term in the upper bound of $|A_{n,m}(t)|$. 
	Recall that the conditional distribution function $F_{m}(\cdot \mid x)$ satisfies the same regularity assumptions imposed in \cite{meinshausen_quantile_2006} on the conditional distribution $F(\cdot \mid X=x)$. 
	As pointed out by the author, it therefore suffices to show that
	\[
	\left|
	\sum_{k=1}^{n} w_k(x)
	\Big(
	\mathbb{1}_{\left\{U_{k,m} \le F_m(y \mid X_{k,m})\right\}}
	-
	\mathbb{1}_{\left\{U_{k,m} \le F_m(y \mid x)\right\}}
	\Big)
	\right|
	\inprob 0,
	\]
	in order to conclude that
	\[
	\sup_{t \in \mathbb{R}} |A_{n,m}(t)|
	\inprob 0.
	\]
	This is the convergence that remains to be established. To this end, for each $k \in \{1,\dots,n\}$ and fixed $y$ such that $y=a_m(x)t+b_m(x),$ define
	\[
	D_{k,m}
	=
	\mathbb{1}_{\left\{U_{k,m} \le F_m(y \mid X_{k,m})\right\}}
	-
	\mathbb{1}_{\left\{U_{k,m} \le F_m(y \mid x)\right\}}.
	\]
	Since \( U_{k,m} \), \( k = 1, \ldots, n \), are uniformly distributed on \([0,1]\), it follows that
	\[
	\mathbb{E}\!\Big(D_{k,m} \Big)
	=
	F_m(y \mid X_{k,m}) - F_m(y \mid x).
	\]
	Using \eqref{resul_wk2_to_0} together with the independence of the \( U_{k,m} \), \( k = 1, \ldots, n \), as \( n \to \infty \),	 we obtain, based on a strategy similar to that of \cite{meinshausen_quantile_2006},	
	\[
	\left|
	\sum_{k=1}^{n} w_k(x)
	\Big(
	\mathbb{1}_{\left\{U_{k,m} \le F_m(y \mid X_{k,m})\right\}}
	-
	\mathbb{1}_{\left\{U_{k,m} \le F_m(y \mid x)\right\}}
	\Big)
	\right|
	\inprob
	\left|
	\sum_{k=1}^{n} w_k(x)
	\Big(F_m(y \mid X_{k,m}) - F_m(y \mid x)\Big)
	\right|.
	\]	
	We now control the deterministic term. By the triangle inequality and by Assumption~\ref{assump_cond_dist_Fm} (mainly the Lipschitz property), we have
	\begin{align*}
		\left|
		\sum_{k=1}^n w_k(x)
		\big(F_m(y \mid X_{k,m}) - F_m(y \mid x)\big)
		\right|
		&\le
		\sum_{k=1}^n w_k(x)
		\left|F_m(y \mid X_{k,m}) - F_m(y \mid x)\right| \\
		&\le
		L \sum_{k=1}^n w_k(x)\,\|X_{k,m} - x\|_2.
	\end{align*}
	By Lemma~\ref{lemme_sum_Wk_Liptchi_cond}, we have
	\[
	\sum_{k=1}^n w_k(x)\,\|X_{k,m} - x\|_2
	\inprob 0.
	\]
	Hence, for every fixed $y\in \mathbb{R}$,
	\[
	\left|
	\sum_{k=1}^n w_k(x)
	\big(F_m(y \mid X_{k,m}) - F_m(y \mid x)\big)
	\right|
	\inprob 0.
	\]	
	Combining the previous results, we conclude that
	\[
	\left|
	\sum_{k=1}^{n} w_k(x)
	\Big(
	\mathbb{1}_{\{U_{k,m} \le F_m(y \mid X_{k,m})\}}
	-
	\mathbb{1}_{\{U_{k,m} \le F_m(y \mid x)\}}
	\Big)
	\right|
	\inprob 0.
	\]	
	Therefore,
	\[
	\sup_{t \in \mathbb{R}} |A_{n,m}(t)|
	\inprob 0.
	\]	
	Finally, combining this result with the convergence of the bias term $B_{n,m}(t)$, we obtain that for every $x \in \mathcal{X}$ and $m=m(n)\to \infty$,
	\[
	\sup_{t \in \mathbb{R}}
	\left|
	\mathbb{F}_n^x(t) - G_{\xi_0(x)}(t)
	\right|
	\inprob 0.
	\]
\end{proof}

This Lemma~\ref{lemme2} allows us to conclude, by Lemma~2.2 in \cite{van_der_vaart_asymptotic_1998}, that for any bounded continuous function $f: \mathbb{R} \to \mathbb{R}$, we have
\begin{equation*}
	\mathbb{P}_n^x[f] \inprob G_{\xi_0(x)}[f] := \int_{\mathbb{R}} f(t) \, dG_{\xi_0(x)}(t).
\end{equation*}
\begin{lem}\label{lemme_sum_Wk_Liptchi_cond}
	Under the assumptions of Lemma~\ref{lemme2}, it holds that, for every fixed $x \in \mathcal{X}$,
	\[
	\sum_{k=1}^n w_k(x)\,\|X_{k,m} - x\|_2
	\inprob 0.
	\]
\end{lem}
The proof of this lemma is given in Appendix~\ref{proof_lemme_sum_Wk_Liptchi_cond}.

\begin{lem}\label{lemma3} 
	Let $x \in \mathcal{X}.$
	Assume that Assumptions~\ref{assump1} to~\ref{assump_cond_dist_Fm} hold and that $\lim\limits_{n\to +\infty} m(n) = +\infty.$ Then, for all upper semi-continuous function $f: [-\infty,+\infty) \to \mathbb{R}$ that is bounded above,
	\[
	\limsup\limits_{n\to +\infty} \mathbb{P}_n^x[f] \leq G_{\xi_0(x)}[f] \quad \text{in probability.}
	\]
\end{lem}

\begin{proof}
	Let $M$ denote the supremum of $f.$ Define $\tilde{f}=M-f,$ which is lower semicontinuous and nonnegative.  
	We have
	\begin{eqnarray}
		\notag
		\mathbb{P}_n^x[f]&=& \sum_{k=1}^{n}w_k(x)\left(M-\tilde{f}\right)(\tilde{Z}_{k,m})\\\notag
		&=& M\sum_{k=1}^{n}w_k(x)-\sum_{k=1}^{n}w_k(x)\tilde{f}(\tilde{Z}_{k,m})\\ \notag
		&=& M-\mathbb{P}_n^x[\tilde{f}] \quad \text{ since $\sum_{k=1}^{n}w_k(x)=1$},
	\end{eqnarray}
	and
	\begin{eqnarray}
		\notag
		G_{\xi_0(x)}[f]&=& \int_{\mathbb{R}}\left(M-\tilde{f}\right)(t)dG_{\xi_0(x)}(t)\\\notag
		&=& M-G_{\xi_0(x)}[\tilde{f}].
	\end{eqnarray} It is therefore sufficient to prove that
	\[\liminf\limits_{n\to +\infty}\mathbb{P}_n^x[\tilde{f}]\geq G_{\xi_0(x)}[\tilde{f}] \quad\text{ in probability.}\tag{$\star$}  \]
	Indeed,
	\begin{eqnarray}\notag
		\limsup\limits_{n \to +\infty}\mathbb{P}_n^x[f]\leq G_{\xi_0(x)}[f]&\Leftrightarrow & \limsup\limits_{n \to +\infty}\left(M-\mathbb{P}_n^x[\tilde{f}]\right)\leq M- G_{\xi_0(x)}[\tilde{f}]\\ \notag
		&\Leftrightarrow & M- \liminf\limits_{n \to +\infty}\mathbb{P}_n^x[\tilde{f}]\leq M- G_{\xi_0(x)}[\tilde{f}] \\ \notag
		&\Leftrightarrow & \liminf\limits_{n \to +\infty}\mathbb{P}_n^x[\tilde{f}]\geq G_{\xi_0(x)}[\tilde{f}].	
	\end{eqnarray}
	To establish $(\star),$ we use the representations
	\[
	\mathbb{P}_n^x[\tilde{f}] = \int_{0}^{1} \tilde{f}\big(F_n^{x\leftarrow}(u)\big)\,du
	\quad \text{and }\quad
	G_{\xi_0(x)}[\tilde{f}] = \int_{0}^{1} \tilde{f}\big(G_{\xi_0(x)}^{\leftarrow}(u)\big)\,du,
	\]
	with $F_n^{x\leftarrow}(u) = \inf\{t\in \mathbb{R} \mid F^x_n(t)\geq u\}.$  
	By Lemma \ref{lemme2}, we know that $\mathbb{F}_n^x(t) \inprob G_{\xi_0(x)}(t)$ for all $t \in \mathbb{R}.$ Furthermore, since $\mathbb{F}_n^x$ is nondecreasing and $G_{\xi_0(x)}$ is continuous, it follows that
	\begin{equation}\label{conv_quantile_proba}
		\mathbb{F}_n^{x\leftarrow}(u) ~\inprob~G_{\xi_0(x)}^{\leftarrow}(u), \quad \forall u \in (0,1).
	\end{equation}
	Let $(u_j)_{j \ge 1}$ be a dense sequence in the interval $(0,1)$. 
	Let $(n_k)$ be an arbitrary subsequence of $(\mathbb{F}_n^{x\leftarrow})_{n \ge 1}$. 
	According to the convergence in probability given by the relation \eqref{conv_quantile_proba}, we can extract from this arbitrary subsequence a diagonal sub-subsequence 
	$(n_{k_\ell})$  such that, almost surely
	\[
	\forall j \ge 1, \quad 
	F_{n_{k_\ell}}^{x\leftarrow}(u_j) \to G_{\xi_0(x)}^{\leftarrow}(u_j) \quad \text{as } \ell \to +\infty.
	\]
	The functions $u \mapsto F_{n_{k_\ell}}^{x\leftarrow}(u)$ are nondecreasing.  
	By density of the set $\{u_j\}$ and continuity of $G_{\xi_0(x)}^{\leftarrow}$, we then obtain that, for almost every $u \in (0,1)$,
	$
	F_{n_{k_\ell}}^{x\leftarrow}(u)\inas G_{\xi_0(x)}^{\leftarrow}(u).
	$

	Since $\tilde{f}$ is lower semi-continuous, we obtain for all $u \in (0,1)$
	\[
	\liminf\limits_{\ell \to +\infty}\tilde{f}\left(\mathbb{F}_{n_{k_\ell}}^{x\leftarrow}(u)\right)\geq \tilde{f}\left(G_{\xi_0(x)}^{\leftarrow}(u)\right) \quad \text{ a.s.},
	\]
	and integrating yields
	\begin{eqnarray}\label{eq1}\notag
		\int_{0}^{1}\tilde{f}\left(G_{\xi_0(x)}^{\leftarrow}(u)\right)du 
		&\leq& \int_{0}^{1}\liminf\limits_{\ell \to +\infty}\tilde{f}\left(\mathbb{F}_{n_{k_\ell}}^{x\leftarrow}(u)\right)du \quad \text{a.s.}\\\notag
		&\leq& \liminf\limits_{\ell \to +\infty}\int_{0}^{1}\tilde{f}\left(\mathbb{F}_{n_{k_\ell}}^{x\leftarrow}(u)\right)du \quad \text{by Fatou's lemma} ,
	\end{eqnarray}
	which gives
	\[
	G_{\xi_0(x)}[\tilde{f}] \leq \liminf\limits_{\ell \to +\infty}\mathbb{P}_{n_{k_\ell}}^x[\tilde{f}] \quad \text{a.s.}
	\]
	We have shown that every subsequence $(n_k)$ admits a sub-subsequence $(n_{k_\ell})$ along which
	$
	\liminf_{\ell \to \infty} \mathbb{P}_{n_{k_\ell}}^x[\tilde{f}] \ge G_{\xi_0(x)}[\tilde{f}] 
	\quad \text{almost surely.}
	$
	By Lemma~4.2 in \cite{kallenberg_foundations_2002}), it then follows that
	\[
	\liminf_{n \to \infty} \mathbb{P}_{n}^x[\tilde{f}] \ge G_{\xi_0(x)}[\tilde{f}] 
	\quad \text{in probability.}
	\]
	This completes the proof of Lemma~\ref{lemma3}.
\end{proof}

The Lemma~\ref{lemma4} and proposition~\ref{propo1}, which we will state below, are very useful for proving the convergence theorem.

\begin{lem}\label{lemma4}
	Let \(x\in\mathcal{X}\). We assume that Assumptions ~\ref{assump1}-\ref{assump_cond_dist_Fm} are satisfied, that \(\xi_0(x)>-1\), and that \(\displaystyle\lim_{n\to\infty}\frac{m(n)}{\log n}=+\infty\). Let \(Y_{k,m} = \ell_{\theta_0}\bigl((Z_{k,m}-b_m(x) )/a_m(x) \bigr)\) and \( f_m(x) = \mathbb{E}[Y_{k,m} \mid X_{k,m} = x] \). 	
	Then
	\[
	\mathbb{P}_n^x[\ell_{\theta_0}]=\sum_{k=1}^n w_k(x)\,Y_{k,m}\inprob G_{\xi_0(x)}[\ell_{\theta_0}]  \quad \text{as $n \to +\infty$,}.
	\]with $\theta_0 = (0,1,\xi_0(x)).$
\end{lem}
The proof of this lemma is thecnical and is given in Appendix \ref{proof_of_lemma5}. 

\begin{proposition}\label{propo1}
	Let $x \in \mathcal{X}$ and let $K \subset \varTheta$ be a compact neighborhood of $\theta_0(x)$, with $\theta_0(x) = (0,1,\xi_0(x))$. Under the conditions of Theorem~\ref{theo_consistency}, we have:
	\[
	\tilde{\theta}_n^K(x) \inprob \theta_0(x) \quad \text{as $n \to +\infty,$}
	\]
	where, $\tilde{\theta}_n^K$ is defined in \eqref{theta_K_n}.
\end{proposition}

\begin{proof}
	Let $x \in \mathcal{X}.$
	From Lemmas 5 and 6 of \cite{dombry_existence_2015}, we have the following results:
	\begin{itemize}
		\item[(i)] For all $\theta \in \varTheta$, we have $G_{\xi_0(x)}[\ell_{\theta}] \le G_{\xi_0(x)}[\ell_{\theta_0}]$, with equality if and only if $\theta = \theta_0$.
		\item[(ii)] For all $B \subset \varTheta$, define $\ell_B(z) = \sup_{\theta \in B} \ell_{\theta}(z)$ for all $z \in \mathbb{R}$.  For all $\theta \in \varTheta$, let $B(\theta, \epsilon)$ denote the open ball in $\varTheta$ centered at $\theta$ with radius $\epsilon>0$. Then,
		\[
		\lim_{\epsilon \to 0} G_{\xi_0(x)}[\ell_{B(\theta,\epsilon)}] = G_{\xi_0(x)}[\ell_{\theta}].
		\]
	\end{itemize}	
	From (i) and (ii), for all $\theta \in K$ with $\theta \neq \theta_0$, there exists $\epsilon_\theta > 0$ such that
	\begin{equation}\label{ineg_import}
		G_{\xi_0(x)}[\ell_{B(\theta, \epsilon_\theta)}] < G_{\xi_0(x)}[\ell_{\theta_0}].
	\end{equation}	
	Fix $\delta>0$ and define $\Delta = \{\theta \in K : \|\theta - \theta_0\| \ge \delta\}$. $\Delta$ is compact and can be covered by the open balls $\{B(\theta, \epsilon_\theta), \theta \in \Delta\}$. Let $B_i = B(\theta_i, \epsilon_{\theta_i})$, $i=1,\dots,p$, be a finite subcover (by the Borel–Lebesgue theorem).	Since $\tilde{L}_n(\theta;x) = \mathbb{P}_n^x[\ell_\theta]$, we have
	
	\begin{eqnarray}
		\notag
		\sup\limits_{\theta \in\Delta} \tilde{L}_n(\theta;x) &\leq&\sup\limits_{\theta \in \cup_{i=1 }^{p}B_i} \tilde{L}_n(\theta;x)\\ \notag
		&\leq& \max_{1\leq i\leq p} \{ \sup \limits_{\theta \in B_i}  \tilde{L}_n(\theta;x)\}\\  \label{key11}
		&\leq&   \max_{1\leq i\leq p} \{ \sup \limits_{\theta \in B_i}\sum_{k=1}^{n}w_k(x)\ell_{\theta}(\tilde{Z}_{k,m})  \}
	\end{eqnarray}
	For each $i \in \{1,\dots,p\}$ and $k \in \{1,\dots,n\}$, we have
	\[
	\forall \theta \in B_i,\quad\ell_\theta(\tilde{Z}_{k,m}) \le \sup_{\theta \in B_i} \ell_\theta(\tilde{Z}_{k,m}),
	\]
	so that	
	\begin{eqnarray}
		\notag
		\sum_{k=1}^{n} w_k(x)\ell_{\theta}(\tilde{Z}_{k,m})& \leq& \sum_{k=1}^{n} w_k(x) \sup \limits_{\theta \in B_i}\ell_{\theta}(\tilde{Z}_{k,m}) \\ \notag
		& \leq& \sum_{k=1}^{n} w_k(x) \ell_{B_i} (\tilde{Z}_{k,m}). 
	\end{eqnarray}
	Hence, $\forall i \in [|1,p|]$
	\begin{eqnarray}\notag
		\sup \limits_{\theta \in B_i}\sum_{k=1}^{n}w_k(x)\ell_{\theta}(\tilde{Z}_{k,m}) & \leq& \sum_{k=1}^{n} w_k(x) \ell_{B_i} (\tilde{Z}_{k,m}) \\ \notag
		& \leq& \mathbb{P}_n^x[\ell_{B_i}]
	\end{eqnarray}
	(\ref{key11}) therefore becomes
	\[
	\sup_{\theta \in \Delta} \tilde{L}_n(\theta;x) \le \max_{1 \le i \le p} \mathbb{P}_n^x[\ell_{B_i}].
	\]
	Since the function $\ell_{B_i}$ is upper semi-continuous and bounded from above, the assumptions of Lemma~\ref{lemma3} are satisfied, and  it follows that
	\begin{eqnarray} \notag
		\limsup_{n \to +\infty}\left(\sup\limits_{\theta \in\Delta} \tilde{L}_n(\theta;x)\right)
		&\leq& \limsup_{n \to +\infty}\left(\max_{1\leq i\leq p} \mathbb{P}_n^x[\ell_{B_i}] \right)\\ \notag
		&\leq&\max_{1\leq i\leq p} \left( \limsup_{n \to +\infty}\mathbb{P}_n^x[\ell_{B_i}] \right) \\ \notag
		&\leq&\max_{1\leq i\leq p} \left( G_{\xi_0(x)}[\ell_{B_i}] \right)  \quad \text{in probability, by Lemma \ref{lemma3}}\\ \notag
		&<& G_{\xi_0(x)}[\ell_{\theta_0} ]   \quad \text{in probability, by equation \eqref{ineg_import}}.
	\end{eqnarray}
	Moreover, for all $\theta \in K$, we have $\tilde{L}_n(\theta;x) \le \sup\limits_{\theta \in K} \tilde{L}_n(\theta;x)$, and in particular
	$
	\tilde{L}_n(\theta_0;x) \le \sup\limits_{\theta \in K} \tilde{L}_n(\theta;x),
	$ so that
	\begin{eqnarray}\notag
		\liminf_{n \to +\infty} \left(\sup\limits_{\theta \in K} \tilde{L}_n(\theta; x)\right)&\geq &  \liminf_{n \to +\infty}\left( \tilde{L}_n(\theta_0,x)\right)  \\ \notag
		&\geq & G_{\xi_0(x)}[\ell_{\theta_0}] \quad\text{ in probability, by Lemma \ref{lemma4}}.
	\end{eqnarray}
	Since $\tilde{\theta}_n^K(x)$ maximizes $\tilde{L}_n$ over $K$, we have
	$$\liminf\limits_{n \to +\infty} \left( \tilde{L}_n(\tilde{\theta}_n^{K}(x))\right)\geq  G_{\xi_0(x)}[\ell_{\theta_0}] > \limsup\limits_{n \to +\infty}\left(\sup\limits_{\theta \in\Delta} \tilde{L}_n(\theta; x)\right).$$
	which implies $\tilde{\theta}_n^{K}(x) \in K \backslash \Delta$ for large $n.$ That is,
	$
	\|\tilde{\theta}_n^K(x) - \theta_0(x)\| < \delta
	$
	for $n$ sufficiently large. Since $\delta$ is arbitrary, this shows that
	\[
	\tilde{\theta}_n^K(x) \inprob \theta_0(x), \quad n \to \infty.
	\]
	
\end{proof}
\begin{proof}[\textbf{Proof of Theorem~\ref{theo_consistency}}]
	Let $x \in \mathcal{X}$ and consider $K \subset \varTheta$ a compact neighborhood of $\theta_0(x)$, and let 
	$\tilde{\theta}_n^K(x) = \big(\tilde{\mu}_n^K(x), \tilde{\sigma}_n^K(x), \tilde{\xi}_n^K(x)\big)$
	be defined as in Proposition~\ref{propo1}. 	
	It suffices to take $	\hat{\theta}_n(x) = \big(a_m(x) \tilde{\mu}_n^K(x) + b_m(x), \, a_m(x) \tilde{\sigma}_n^K(x), \, \tilde{\xi}_n^K(x)\big)$  to obtain the results (\ref{equation_existence}) and (\ref{equation_consistence}) of Theorem~\ref{theo_consistency}. Indeed, with $\hat{\theta}_n(x) = (\hat{\mu}_n(x), \hat{\sigma}_n(x), \hat{\xi}_n(x))$ considered,
	we have
	\[
	\left( \frac{\hat{\mu}_n(x) - b_m(x) }{a_m(x) }, \frac{\hat{\sigma}_n(x)}{a_m(x) }, \hat{\xi}_n(x) \right)
	= \left(\tilde{\mu}_n^K(x), \tilde{\sigma}_n^K(x), \tilde{\xi}_n^K(x) \right)
	= \tilde{\theta}_n^K(x).
	\]
	By Proposition~\ref{propo1}, $\tilde{\theta}_n^K(x) \inprob \theta_0(x)$, which is equivalent to
	\[ 
	\hat{\xi}_n(x) \inprob \xi_0(x), \quad 
	\frac{\hat{\mu}_n(x)-b_m(x) }{a_m(x) } \inprob 0, \quad 
	\frac{\hat{\sigma}_n(x)}{a_m(x) } \inprob 1 \quad \text{as } n \to \infty,
	\] 
	and thus (\ref{equation_consistence}) holds. 	
	Moreover, since  $\tilde{\theta}_n^K(x)=\arg\max\limits_{\theta\in K}\tilde{L}_n(\theta;x),$  
	$\tilde{L}_n(\theta;x)$ attains a local maximum at $\tilde{\theta}_n^K(x)$ whenever $\tilde{\theta}_n^K(x) \in \text{Int}(K)$ (see proof of Theorem 2 in \cite{dombry_existence_2015}). As $\theta_0(x) \in \text{Int}(K)$, we have $\tilde{\theta}_n^K(x) \in \text{Int}(K)$ for $n$ sufficiently large. Therefore, $\tilde{L}_n(\theta;x)$ admits a local maximum at $\tilde{\theta}_n^K(x)$ for large $n$, and by Lemma~\ref{lemma1}, $\hat{\theta}_n(x)$ is a maximum likelihood estimator for $n$ sufficiently large. This proves (\ref{equation_existence}) and completes the proof of Theorem~\ref{theo_consistency}.
	
\end{proof}

\section{Simulation study}\label{section_simulation}
\subsection{Experiment 1}\label{expe1}
The aim of this first simulation study is to empirically demonstrate the convergence of the proposed weighted maximum likelihood estimator and to illustrate the theoretical results established in Section~\ref{section_mainContribution}. To this end, we independently generate a training sample
$
\{(X_i,Y_i)\}_{i=1}^{N},
$
according to the conditional model
\begin{equation}\label{cond_modelExp1}
	Y=\mu(X)+\sigma(X)\varepsilon,
	\qquad
	\varepsilon\sim\mathrm{GEV}(0,1,\xi),
\end{equation}
where $\sigma(X)>0.$
The covariate vector
\[
X_i=(X_i^{(1)},\ldots,X_i^{(p)})
\]
is independently generated from the uniform distribution on $[-1,1]^p$, with $p=20$. 
Consequently, conditional on $X=x$, the response variable follows a generalized extreme value (GEV) distribution with parameters $\mu(x)$, $\sigma(x)$, and $\xi$, namely
$
Y\,|\,X=x
\sim
\mathrm{GEV}\bigl(\mu(x),\sigma(x),\xi\bigr).
$
We consider the following model specification throughout this simulation study:
\[
\mu(x)=1+\mathbb{1}_{\{x^{(1)}\geq0\}},
\qquad 
\sigma(x)=4-\left((x^{(1)})^2-(x^{(2)})^2\right).
\]
Only the first two covariates contribute to the conditional distribution, while the remaining variables are treated as noise variables. Since the data are generated directly from a conditional GEV distribution, no block maxima construction is required in this first experiment, which is equivalent to considering a block size equal to one. The influence of the block size will be investigated in a second experiment. To evaluate the performance of the proposed estimator, we generate an independent test sample
$
\{x_i\}_{i=1}^{M},
$
using a Halton sequence, which provides a nearly uniform coverage of the covariate space. The estimator is then evaluated for training sample sizes ranging from $100$ to $4000$ observations.
For each parameter
$
\alpha\in\{\mu,\sigma,\xi\},
$
we consider the integrated variance defined as
\[
IVar(\alpha)
=
\frac{1}{M}
\sum_{i=1}^{M}
\mathrm{Var}\!\left(\widehat{\alpha}(x_i)\right),
\]
where $\mathrm{Var}(\widehat{\alpha}(x_i))$ denotes the empirical variance of the estimator at the prediction point $x_i$, computed over different Monte Carlo replications.
We also consider the integrated mean squared error (denoted by MISE), defined as
\[
MISE(\alpha)
=
IVar(\alpha)+IBias^2(\alpha).
\]
where $IBias$  denotes the integrated squared bias, given by
$IBias^2(\alpha)=\frac{1}{M}\sum_{i=1}^{M}\left(\mathbb{E}\!\left[\widehat{\alpha}(x_i)\right]-\alpha(x_i)\right)^2.
$
\begin{figure}[h]
	\centering
	\includegraphics[width=0.99\linewidth, height=0.5\textheight]{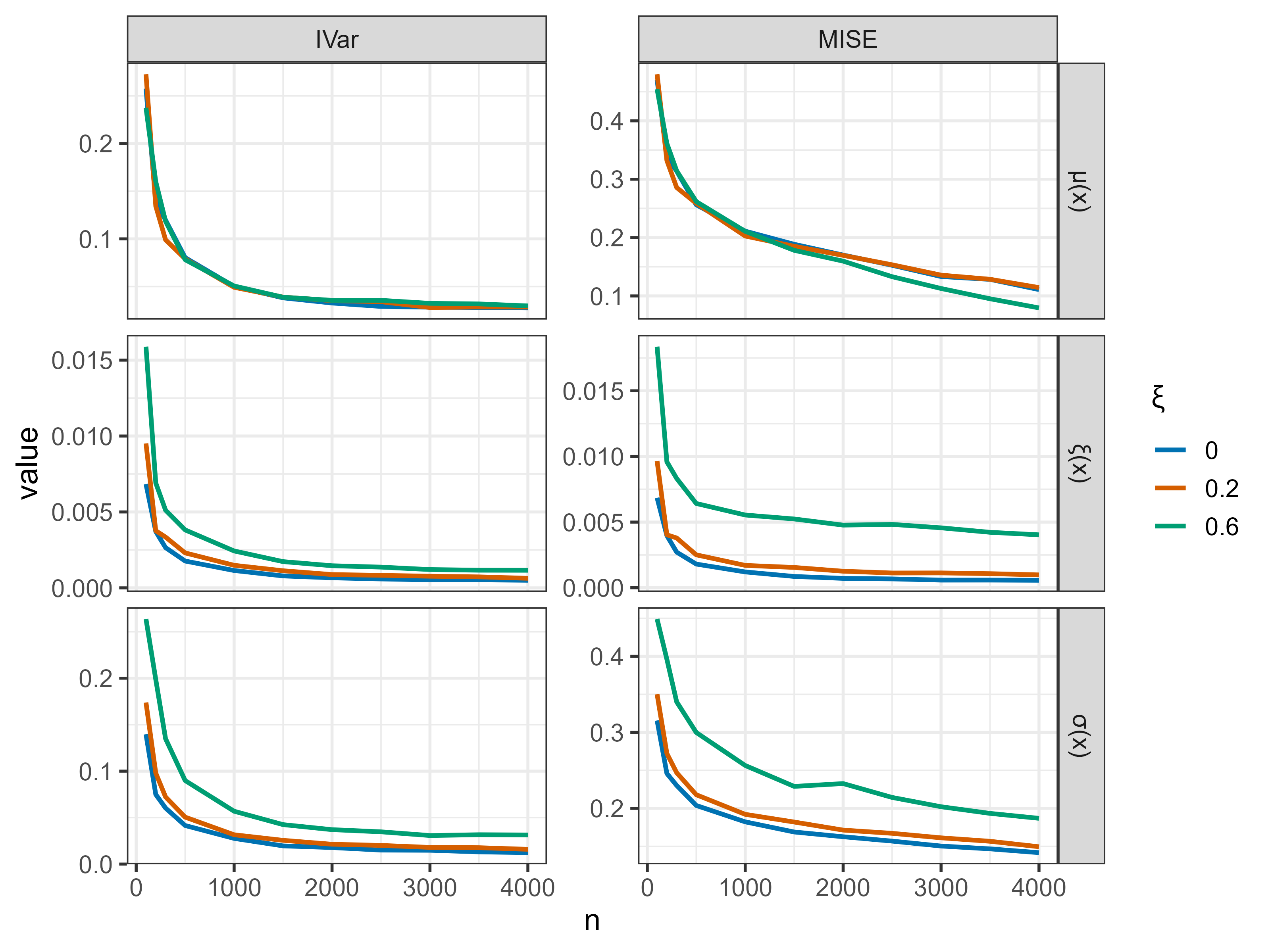}
	\caption{Integrated variance (IVar) and MISE as functions of the sample size $n.$	}
	\label{fig:figueexperience1}
\end{figure}
 Throughout this study, we use the default tuning parameters of the generalized random forest algorithm, as recommended by \cite{athey_generalized_2019}, namely $min.node.size=5$ and $num.trees=2000$.
Figure~\ref{fig:figueexperience1} shows the evolution of the integrated variance and the MISE of the proposed estimators as the sample size increases for three values of the shape parameter,
$
\xi\in\{0,\;0.2,\;0.6\}.
$
The results provide clear empirical evidence of the convergence of the proposed estimator. For all three model parameters, the integrated variance decreases rapidly as the sample size increases, resulting in a steady decline of the MISE. As expected, the estimation becomes more challenging as the shape parameter increases, reflecting the greater difficulty of estimating heavier-tailed distributions. Despite this increased difficulty, both criteria decrease systematically with the sample size, thereby confirming the stability and consistency of the proposed estimator, in agreement with the asymptotic results established in Section~\ref{section_mainContribution}.
\begin{figure}[h]
	\centering
	\includegraphics[width=0.99\linewidth, height=0.5\textheight]{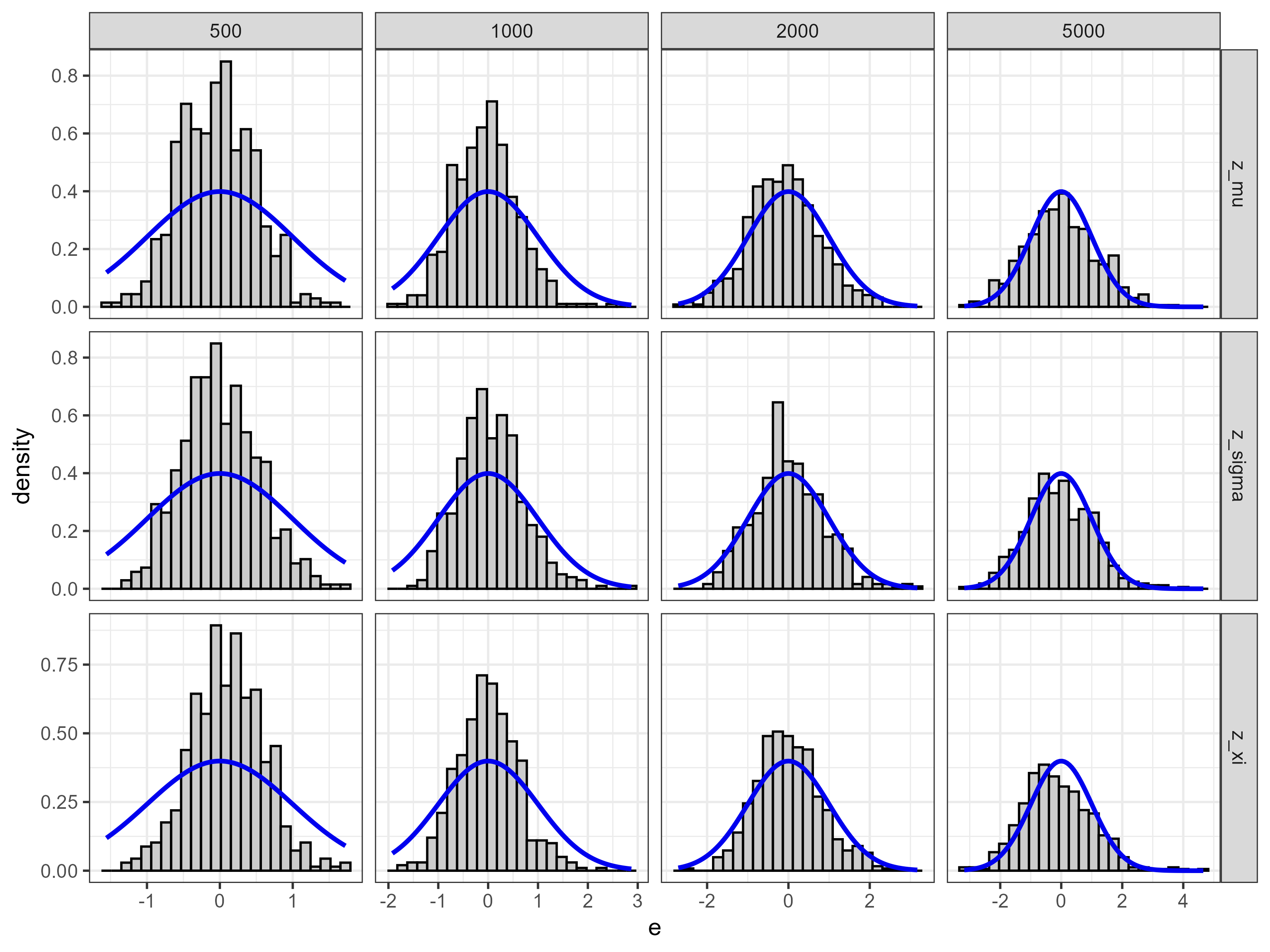}
	\caption{Histograms of the standardized estimation errors
		$z$
		for the estimators of $\mu(x_0)$, $\sigma(x_0)$ and $\xi(x_0)$, with
		$x_0=(0.5,\ldots,0.5)$ and different sample sizes. The blue curve represents the density of the standard normal distribution $\mathcal{N}(0,1)$.}
	\label{fig:figuenormassymp}
\end{figure}

To empirically illustrate the asymptotic normality of the proposed estimators, we consider the same setting as that described in \eqref{cond_modelExp1}, with the shape parameter fixed at $\xi=0.2$. The location function $\mu(x)$ and the scale function $\sigma(x)$ are identical to those used in the previous experiment. We consider $p=5$ and a fixed point in the covariate space,
$
x_0=(0.5,\ldots,0.5).
$
For each sample size
$
n\in\{500,1000,2000,5000\},
$
we generate $L$ independent training samples. At each replication, the proposed model is fitted to the simulated data, and the parameters $\mu(x_0)$, $\sigma(x_0)$, and $\xi(x_0)$ are estimated at the point $x_0$. For each parameter $\alpha\in\{\mu,\sigma,\xi\}$, we compute the normalized errors
\[
e_l=\sqrt{n}\left(\widehat{\alpha}^{(l)}(x_0)-\alpha(x_0)\right),
\qquad l=1,\ldots,L,
\]
where $\alpha(x_0)$ denotes the true value of the parameter $\alpha$ at the point $x_0$, while $\widehat{\alpha}^{(l)}(x_0)$ denotes the estimate of this parameter at the same point $x_0$, obtained from the $l$-th sample. Since the asymptotic variance is unknown, these errors are then centered and standardized using their empirical mean $\bar{e}$ and empirical standard deviation $s_e$, leading to the variables
\[
z_l=\frac{e_l-\bar{e}}{s_e}.
\]
Figure~\ref{fig:figuenormassymp} presents, for each of the three parameters and for the different sample sizes, the histograms of the standardized variables $z_l$, with the density of the standard normal distribution $\mathcal{N}(0,1)$ superimposed (shown in blue).
For the smallest sample size ($n=500$), the empirical distributions still exhibit a slight asymmetry and deviate moderately from the theoretical normal density. However, as the sample size increases, the histograms become progressively more symmetric, and their shape approaches that of the normal distribution. This convergence is particularly pronounced for $n=2000$ and $n=5000$, where the empirical distributions of the estimators of $\mu(x)$, $\sigma(x)$, and $\xi(x)$ closely reproduce the characteristic bell shape of a Gaussian distribution. This behavior is consistently observed for all three parameters.  These results provide empirical evidence of an approximately Gaussian behavior of the normalized estimators for increasing sample sizes.

\subsection{Experiment 2}
In this second experiment, we evaluate the ability of the proposed method to address the extreme quantile regression challenges discussed in the introduction of this work. To this end, we consider the conditional model
$
Y \mid X=x \sim \gamma(x)\,\mathcal{T}_{\nu(x)},
$
where
\[
\gamma(x)=1+2\pi\varphi(2x_1,2x_2),
\qquad
\nu(x)=3+\frac{7}{1+\exp(4x_1+1.2)},
\]
with $\varphi$ denoting the density of the centered bivariate normal distribution with unit variance and correlation coefficient equal to $0.75$. The covariate vector
$
X=(X^{(1)},\ldots,X^{(p)})
$
is generated from the uniform distribution on $[-1,1]^p$, and a training sample
$
\{(X_i,Y_i)\}_{i=1}^{N}
$
is then simulated according to this model.
Within the block maxima framework, the choice of the block size is a crucial step, as it directly determines the trade-off between the approximation bias induced by the GEV model and the variance of the resulting estimator. To date, there is no universally accepted procedure for selecting an optimal block size. Several methods have been proposed in specific settings, notably based on statistical tests (see, for example, \cite{vidagbandji_generalized_2025} and \cite{,vidagbandji_penalized_2026-1}). In practice, however, the block size is most often chosen according to natural time units, such as years, seasons, or months, depending on the temporal resolution of the data.
\begin{figure}[h!]
	\centering
	\includegraphics[width=0.99\linewidth, height=0.45\textheight]{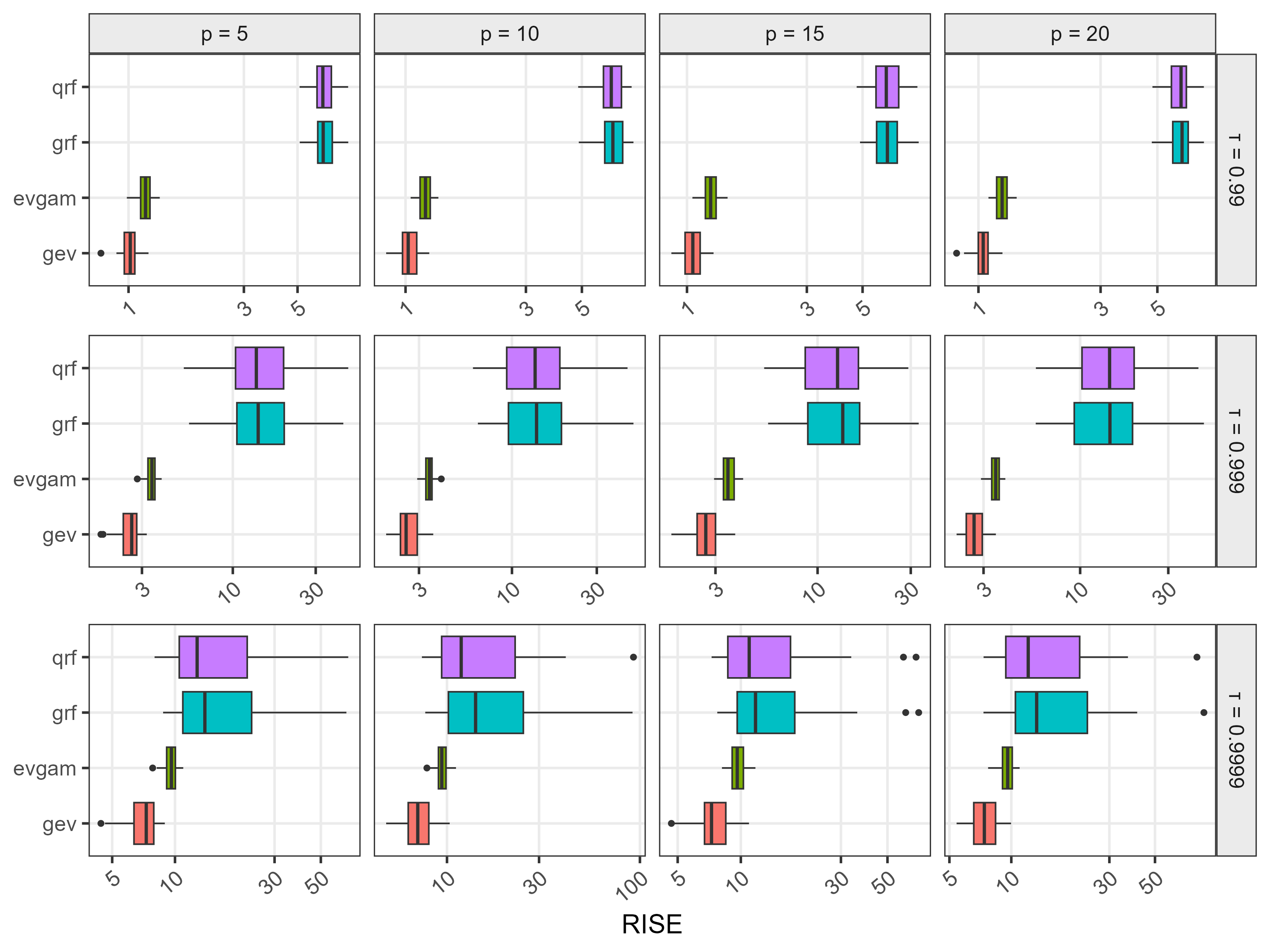}
	\caption{Boxplots of the RISE over 50  replications for different quantile regression methods across probability levels and covariate dimensions.}
	\label{fig:figueexperience21}
\end{figure}

In this study, we rely on the convergence result established in Section~\ref{section_mainContribution}, which requires that
$
\lim_{n\rightarrow\infty}\frac{m(n)}{\log n}=+\infty.
$
Accordingly, we set the block size to
$
m=\lceil n^{\frac{1}{3}}\rceil,
$
where $\lceil\cdot\rceil$ denotes the ceiling function. The number of block maxima used for inference is fixed at $n=2000$, resulting in an initial sample of size
$
N=m\times n.
$
The performance of the proposed method is evaluated on an independent test sample
$
\{x_i\}_{i=1}^{M},
$
generated using a Halton sequence over $[-1,1]^p$, with $M=200$. We compare our approach with several existing quantile regression methods, namely Quantile Regression Forest (qrf) \cite{meinshausen_quantile_2006}, Generalized Random Forest (grf) \cite{athey_generalized_2019}
, and the generalized additive model for extreme values (evgam) proposed by \cite{youngman_evgam_2022}.
 For all forest-based methods, we use the default tuning parameters recommended by \cite{athey_generalized_2019}, namely $num.trees=2000$ and $min.node.size=5$.  
 For an estimated conditional quantile function
 $
 x\mapsto\widehat{\mathcal{Q}}_{x}(\tau),
 \quad \tau\in(0,1),
 $
 we evaluate its predictive performance using the empirical Integrated Squared Error (ISE) computed over the test sample,
 $
 \mathrm{ISE}
 =
 \frac{1}{M}
 \sum_{i=1}^{M}
 \left(
 \widehat{\mathcal{Q}}_{x_i}(\tau)
 -
 \mathcal{Q}_{x_i}(\tau)
 \right)^2,
 $
 where $\mathcal{Q}_{x}(\tau)$ denotes the true conditional quantile function of the data-generating model and $M$ is the size of the test sample. As a performance criterion, we report the Root Integrated Squared Error (RISE), defined as
 $\mathrm{RISE}=\sqrt{\mathrm{ISE}}.$
 \begin{figure}[h]
 	\centering
 	\includegraphics[width=0.99\linewidth, height=0.5\textheight]{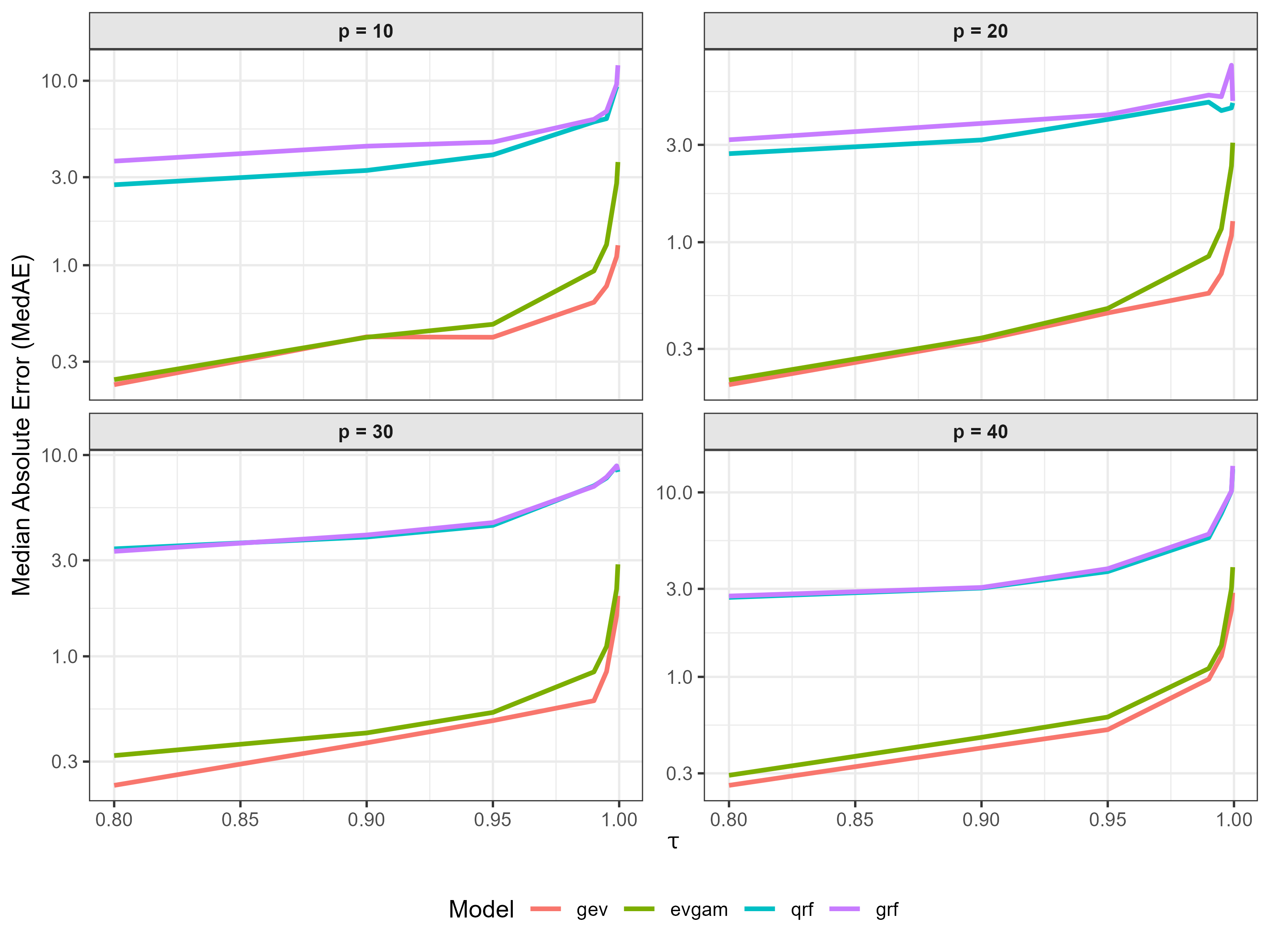}
 	\caption{Median Absolute Error of competing quantile regression methods across probability levels and covariate dimensions.}
 	\label{fig:figueexperience22}
 \end{figure}
 
 Figure~\ref{fig:figueexperience21} presents the boxplots of the RISE, computed over $50$ independent Monte Carlo replications, for the different methods under comparison. The results are reported for three extreme quantile levels, $\tau\in\{0.99,\,0.999,\,0.9999\}$, and four covariate dimensions ($p\in \{5,10,15,20\}$). This figure clearly highlights the superiority of the proposed method (denoted  gev) over the competing approaches. It outperforms both evgam, which is also based on Extreme Value Theory and generalized additive models, and the methods relying solely on statistical learning, namely grf and qrf. Although evgam yields substantially better results than qrf and grf, highlighting the benefit of a model specifically designed for extreme values, it remains consistently less accurate than the proposed approach. This improvement can be attributed to the localized weighted maximum likelihood estimation, which allows the parameters of the GEV distribution to adapt more flexibly to the covariates than the additive specification adopted by evgam. Finally,  we observe that the performance gap between the proposed method and the competing approaches becomes more pronounced as the target quantile level increases. This result shows that incorporating Extreme Value Theory is particularly beneficial for estimating very high conditional quantiles, a setting in which classical quantile regression methods quickly reach their limitations. Overall, these results confirm that the proposed method provides more accurate and more stable estimates of extreme conditional quantiles.

Finally, we consider the Median Absolute Error (MedAE) as an additional performance measure. It is defined as the median of the absolute differences between the estimated and true conditional quantiles,
$
\left|
\widehat{\mathcal{Q}}_{x_i}(\tau)-\mathcal{Q}_{x_i}(\tau)
\right|,
$
computed over the test sample $\{x_i\}_{i=1}^{M}$. Figure~\ref{fig:figueexperience22} displays the evolution of the MedAE as a function of the probability level for different covariate dimensions, $p\in\{10,20,30,40\}$. The conclusions are consistent with those obtained using the RISE criterion. As expected, the MedAE increases as the target probability level approaches one, illustrating the increasing difficulty of estimating more extreme conditional quantiles. Nevertheless, the proposed method consistently outperforms evgam, grf, and qrf across all simulation settings. This superiority is maintained even in high-dimensional settings, demonstrating the ability of the proposed approach to accurately estimate extreme conditional quantiles when the underlying quantile function is highly nonlinear and the covariate space is of high dimension.

\section{Real data analysis study}\label{section_realData}  

We applied our model using value-weighted return data from $10$ industry portfolios, constructed from all firms listed on the NYSE, AMEX, and NASDAQ. This dataset, obtained from the Kenneth French Data Library, contains 26,274 observations, corresponding to daily returns from July $1926$ to June $2026$ for the following ten industry portfolios: Consumer Nondurables (denoted NoDur), Consumer Durables (denoted Durbl), Manufacturing (denoted Manuf), Energy (denoted Enrgy), Business Equipment (denoted HiTec), Telephone and Television Transmission (denoted Telcm), Shops, Healthcare (denoted Hlth), Utilities(denoted Utils), and Other. For further details on these economic sectors and access to the value-weighted return data, the reader is referred to the Kenneth French Data Library\footnote{See \url{https://mba.tuck.dartmouth.edu/pages/faculty/ken.french/Data_Library/det_10_ind_port.html} for more detail.}.

We are interested in extreme losses. To this end, we transform the simple returns into log-returns and multiply them by $(-1)$ to obtain a measure of losses. If $r_{t,j}$ denotes the return of portfolio $j$ at time $t$, we consider the transformation
$X_{t,j}=-\log\left(\frac{r_{t,j}}{100}+1\right).
$ The Ljung–Box tests applied to these logarithmic losses lead to the rejection of the null hypothesis of no autocorrelation for all the portfolios considered. These results, presented in Appendix~\ref{Appendix_diagnostic}, suggest that the logarithmic loss series exhibit conditional temporal dependence. Consequently, the i.i.d. assumption required by the theoretical framework of our model is not satisfied by the raw data. Filtering both the temporal dependence  is therefore necessary before applying the extreme quantile regression model. In our application, we consider a more general specification that accounts simultaneously for dependence in the conditional mean and conditional variance. For each portfolio $j$, we specify an ARMA(2,2)-GARCH(1,1) model:
\begin{equation}
	X_{t,j}
	=
	\mu_{t,j}
	+
	\sigma_{t,j}\varepsilon_{t,j},
\end{equation}
where the mean component follows an ARMA(2,2) model,
\begin{equation*}
	\mu_{t,j}
	=
	c_j
	+\phi_{1,j}X_{t-1,j}
	+\phi_{2,j}X_{t-2,j}
	+\theta_{1,j}\eta_{t-1,j}
	+\theta_{2,j}\eta_{t-2,j},
\end{equation*}
with
\begin{equation*}
	\eta_{t,j}
	=
	\sigma_{t,j}\varepsilon_{t,j},
\end{equation*}
while the conditional variance follows a GARCH(1,1) process,
\begin{equation*}
	\sigma_{t,j}^{2}
	=
	\omega_j
	+\alpha_j\eta_{t-1,j}^{2}
	+\beta_j\sigma_{t-1,j}^{2}.
\end{equation*}
The parameters are estimated marginally for each of the ten portfolios. The associated standardized residual is then defined by
\begin{equation*}
	\widehat{\varepsilon}_{t,j}
	=
	\frac{X_{t,j}-\widehat{\mu}_{t,j}}
	{\widehat{\sigma}_{t,j}}.
\end{equation*}
The Ljung--Box tests are applied to the standardized marginal residual series for lags ranging from $1$ to $3$. The results are summarized in Table~\ref{tab:diagnostic_tests} in Appendix~\ref{Appendix_diagnostic}. For each of the portfolios considered, the test fails to reject the null hypothesis of no serial autocorrelation. These results therefore suggest that the standardized residuals do not exhibit significant temporal dependence.
\begin{figure}
	\centering
	\includegraphics[width=0.9999\linewidth, height=0.5\textheight]{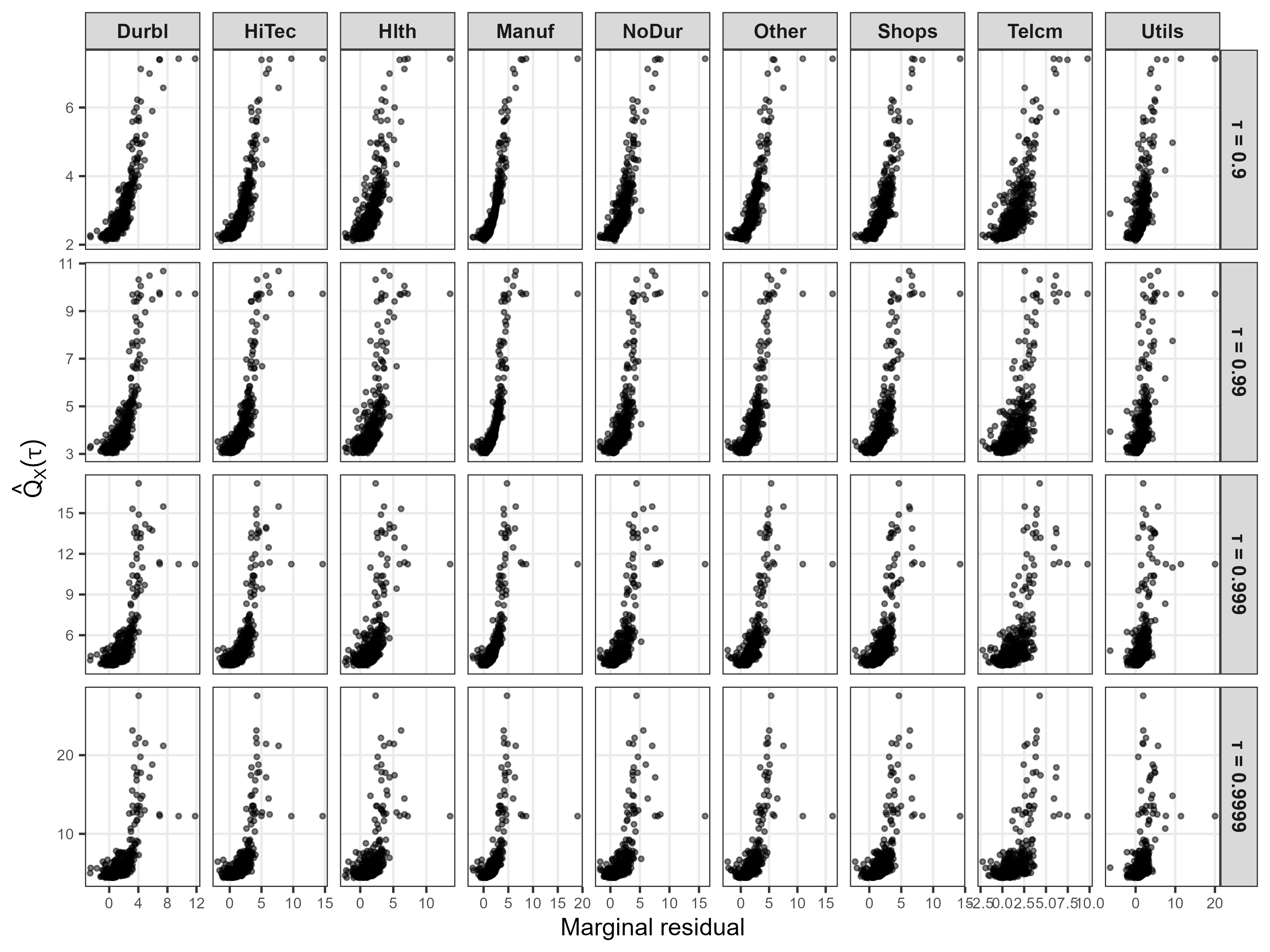}
	\caption{Relationship between estimated conditional quantiles and Portfolio Residuals across Sectors.}
	\label{fig:plotquantportfolio}
\end{figure}

The objective of this step is not to model the temporal dynamics of the losses, but rather to obtain marginal series whose serial dependence is sufficiently reduced to allow for the analysis of the block maxima approach. We therefore use these residuals in the remainder of our analysis. We consider the residuals associated with the Energy sector portfolio as the response variable,
$
Y_t=\widehat{\varepsilon}_{t,\mathrm{Enrgy}},
$ and the residuals of the other nine sector portfolios as covariates, collected in the vector
$
X_t=
\left(
\widehat{\varepsilon}_{t,\mathrm{NoDur}},
\widehat{\varepsilon}_{t,\mathrm{Durbl}},
\widehat{\varepsilon}_{t,\mathrm{Manuf}},
\widehat{\varepsilon}_{t,\mathrm{HiTec}},
\widehat{\varepsilon}_{t,\mathrm{Telcm}},
\widehat{\varepsilon}_{t,\mathrm{Shops}},
\widehat{\varepsilon}_{t,\mathrm{Hlth}},
\widehat{\varepsilon}_{t,\mathrm{Utils}},
\widehat{\varepsilon}_{t,\mathrm{Other}}
\right).
$ The filtered observations \((Y_t,X_t)\) are then partitioned into monthly blocks. For each block \(B_k\), with \(k=1,\ldots,1200\), the response variable is defined as the maximum of the standardized residuals associated with the Energy sector losses:
$
Z_k=\max_{t\in B_k}\widehat{\varepsilon}_{t,\mathrm{Enrgy}},
$ where \(B_k\) denotes the \(k\)-th monthly block. The explanatory variables consist of the standardized residuals of the other sector portfolios, collected in \(X_k\), observed on the same day as the maximum residual of the Energy portfolio. 
Thus, we obtain the sample of monthly maxima with concomitant covariates,
$
\{(Z_k,X_k),\,k=1,\ldots,1200\}.$

First, we train the model on \(60\%\) of the observations \((Z_k,X_k)_{k=1,2,\ldots}\) and estimate the conditional quantiles on the remaining \(40\%\). Figure \ref{fig:plotquantportfolio} illustrates the evolution of the estimated conditional quantiles as a function of the marginal residuals of the different sectors. The figure highlights a positive and nonlinear relationship between the marginal residuals and the estimated conditional quantiles across all sectors. As \(\tau\) close to \(1\), both the estimated quantile values and their dispersion increase markedly. Differences across sectors also emerge, reflecting heterogeneity in tail behavior. These results suggest that the dependence between the marginal residuals and the estimated conditional quantiles becomes more pronounced in the extreme regions.

Finally, we quantitatively evaluate the predictive performance of the model. To this end, we divide the monthly maxima data into five subsamples (folds) in order to implement a 5-fold cross-validation procedure. For each fold \(k\), the model is fitted using the observations belonging to the other four folds, while the observations in fold \(k\) are used to evaluate its predictive performance. The prediction error is assessed using the metric proposed by \cite{wang_estimation_2013}, defined as
$
R_{n'}\left(\widehat{\mathcal{Q}}(\tau)\right)
=
\frac{
	\displaystyle\sum_{i=1}^{n'}
	\mathbb{1}_{\{Y_i<\widehat{\mathcal{Q}}_{X_i}(\tau)\}}
	-n'\tau
}{
	\sqrt{n'\tau(1-\tau)}
},$ where \(\widehat{\mathcal{Q}}_{x}(\tau)\) denotes the estimated conditional quantile at level \(\tau\), evaluated on the test sample \(\{(X_i,Y_i)\}_{i=1,\ldots,n'}\). To reduce the variability associated with the random partitioning of the data, the five-fold cross-validation procedure is repeated ten times using different partitions. For each repetition, we compute the absolute value of the prediction error obtained for each fold. The final predictive performance of each model is then evaluated by averaging these absolute errors over all folds and across the ten repetitions. The results of this performance evaluation, together with those obtained from the evgam, grf, and qrf models, are reported in Table\ref{tab:quantile_errors}. As shown in the table, the gev model achieves the best predictive performance among the models considered.
\begin{table}[htbp]	
	\renewcommand{\arraystretch}{1.5}
	
	\centering

	\begin{tabular}{lcccc}
		\toprule
		Probability level ($\tau$) & $\mathrm{gev}$ &$\mathrm{evgam}$ &
$\mathrm{grf}$ & $\mathrm{qrf}$ \\
		\midrule
		0.8  & 4.50  & 4.65  & 4.49  & 4.69  \\
		0.9  & 10.5  & 10.9  & 11.1  & 11.1  \\
		0.99 & 39.0  & 41.7  & 40.8  & 40.6  \\
		0.995  & 66.0  & 69.7  & 67.9  & 68.5  \\
		0.999  & 135   & 142   & 139   & 140   \\
		0.9995 & 192   & 203   & 198   & 200   \\
		0.9999 & 411   & 446   & 434   & 439   \\
		\bottomrule
	\end{tabular}
		\caption{Comparison of Mean Absolute Quantile Estimation Errors from Repeated 5-Fold Cross-Validation Using the $R_{n'}(\widehat{\mathcal{Q}}(\tau))$ Metric.}
			\label{tab:quantile_errors}
\end{table}

\section{Conclusion}
In this paper, we introduced a flexible and theoretically grounded framework for estimating extreme conditional quantiles in high-dimensional and complex non-linear settings. By uniting the theoretical rigor of extreme value theory, specifically through the block maxima approach, with the non-parametric predictive power of generalized random forests, our method effectively overcomes two fundamental barriers in extremal quantile regression: modeling complex covariate interactions without restrictive parametric assumptions, and accurately extrapolating into sparse tail regions. The primary contribution of this work is theoretical. Under a conditional max-domain of attraction assumption and suitable localization and regularity properties of the forest weights, we established the existence and consistency of the local weighted maximum likelihood estimator for the GEV parameters and the corresponding extreme conditional quantiles. This theoretical foundation fills a critical gap in the literature on statistical learning methods applied to block maxima models, providing formal guarantees for localized unpenalized estimators. 
Through extensive simulation studies, we have empirically validated our theoretical results, confirming both the consistency and the approximate asymptotic normality of the proposed estimators as the sample size increases. Comparative experiments have shown that our method consistently outperforms other approaches, including classical quantile regression forests (\cite{meinshausen_quantile_2006}), generalized random forests (\cite{athey_generalized_2019}), and generalized additive models for extreme values (\cite{youngman_evgam_2022}), particularly at very high probability levels and in high-dimensional predictor spaces. Furthermore, an empirical application to sector portfolio losses in U.S. stock markets has illustrated the practical value of our approach in capturing the nonlinear dependencies of distribution tails and dynamic risk structures in real financial data.

Several promising avenues for further research emerge from this work. A natural and important theoretical extension is to formally prove the asymptotic normality of the forest-weighted maximum likelihood estimator. Establishing such asymptotic properties would allow us to construct valid confidence intervals and perform formal hypothesis tests for extreme conditional quantiles. Although our theoretical framework assumes an appropriate growth rate for the block size $m(n)$, the development of adaptive, data-driven selection procedures for $m,$ allowing for a balance between approximation bias and estimation variance in the presence of covariates, remains an open problem whose solution poses a computational challenge. Finally, extending the proposed methodology to account for dependent or non-stationary data, such as spatiotemporal environmental processes or high-frequency financial time series, would broaden its scope of application to risk assessment in climate science and quantitative finance.

\section*{Acknowledgements}
The authors gratefully thank Clément Dombry for insightful discussions and valuable comments that helped improve this work.

\appendix
\section*{Appendix}
\section{Proof of lemma~\ref{lemme_sum_Wk_Liptchi_cond}}\label{proof_lemme_sum_Wk_Liptchi_cond}
Under the assumptions of Lemma~\ref{lemme2}, consider each tree 
$b = 1,\dots,B$. 
For a fixed $x \in \mathcal{X}$, let $R_b(x)$ denote the leaf of tree $b$ 
containing $x$, and let $|R_b(x)|$ be the number of observations in that leaf. 
The diameter of $R_b(x)$ is defined as
\[
\operatorname{diam}(R_b(x))
=
\sup\{\|y - x\|_2 : y \in R_b(x)\}.
\] 
By construction of the random forest, for $x\in\mathcal{X},$ the weights assigned to the 
observation $X_k$ can be written as
\[
w_k(x)
=
\frac{1}{B}
\sum_{b=1}^B
\frac{\mathbb{1}_{\{X_{k,m} \in R_b(x)\}}}{|R_b(x)|}.
\]
Therefore,
\[
\sum_{k=1}^n w_k(x)\|X_{k,m} - x\|_2
=
\frac{1}{B}
\sum_{b=1}^B
\frac{1}{|R_b(x)|}
\sum_{X_{k,m} \in R_b(x)}
\|X_{k,m} - x\|_2.
\]
For each tree $b$, every point in the leaf $R_b(x)$ lies at a distance 
from $x$ bounded above by the diameter of the leaf. Hence, for all 
$b \in \{1,\dots,B\}$,
\[
\frac{1}{|R_b(x)|}
\sum_{X_{k,m} \in R_b(x)}
\|X_{k,m} - x\|_2
\le
\operatorname{diam}(R_b(x)).
\]
It follows that
\[
\sum_{k=1}^n w_k(x)\|X_{k,m} - x\|_2
\le
\frac{1}{B}
\sum_{b=1}^B
\operatorname{diam}(R_b(x)).
\]
Assumption~\ref{assump_diam_convProba} ensures that, for each tree $b$, 
$\operatorname{diam}(R_b(x))
\xrightarrow{p}
0
\quad \text{as } s \to \infty,
$
and therefore as $n \to \infty$. 
Since the number of trees $B$ is fixed, the average diameter also converges 
to zero in probability. Indeed, for any $\varepsilon > 0$,
\[
\mathbb{P}\!\left(
\frac{1}{B}
\sum_{b=1}^B
\operatorname{diam}(R_b(x))
>
\varepsilon
\right)
\le
\sum_{b=1}^B
\mathbb{P}\!\left(
\operatorname{diam}(R_b(x)) > \varepsilon
\right)
\longrightarrow 0.
\]
Consequently,
\[
\sum_{k=1}^n w_k(x)\|X_{k,m} - x\|_2
\inprob 0.
\]

\section{Proof of lemma \ref{lemma4}}\label{proof_of_lemma5}

The proof of  lemma \ref{lemma4}, which is somewhat technical, is facilitated by Lemma~\ref{lemma6} below. 
\begin{lem}\label{lemma6}
	Let \(x \in \mathcal{X}\). Suppose that the assumptions of Lemma \ref{lemma4} are satisfied and let \(T_{k,m}= Y_{k,m} - \mathbb{E}[Y_{k,m}\mid X_{k,m}]\). Then
	\[
	\sum_{k=1}^n w_k(x)\,T_{k,m}\xrightarrow{\mathbb{P}} 0.
	\]
\end{lem}
\begin{proof}
	Let \(T_n = \sum_{k=1}^n w_k(x)T_{k,m}\). By the property of conditional expectation, we have
	\begin{eqnarray}\notag
		\mathbb{E}[T_{k,m}^2\mid X_{k,m}] &=&\mathbb{E}\left[Y_{k,m}^2+[\mathbb{E}(Y_{k,m}\mid X_{k,m})]^2-2Y_{k,m}\mathbb{E}(Y_{k,m}\mid X_{k,m})\Big| X_{k,m}\right]\\\notag
		&=&\mathbb{E}\left[Y_{k,m}^2\Big| X_{k,m}\right]+ \mathbb{E}\left[[\mathbb{E}(Y_{k,m}\mid X_{k,m})]^2\Big| X_{k,m}\right]-2\left[\mathbb{E}(Y_{k,m}\Big| X_{k,m})\right]^2\\\notag
		&\leq&\mathbb{E}\left[Y_{k,m}^2\Big| X_{k,m}\right]-\mathbb{E}\left[Y_{k,m}\Big| X_{k,m}\right]^2\\\notag
		&\leq&\mathbb{E}\left[Y_{k,m}^2\Big| X_{k,m}\right]\\\notag
		&\leq&C\quad \text{by equation (\ref{Hypothese_H}) of  Assumption~\ref{assump_cond_dist_Fm}.}
	\end{eqnarray}	
	Therefore, since the $(T_{k,m},X_{k,m})$ are independent for $k=1,\cdots,n$, we have
	\begin{eqnarray}\notag
		\mathbb{E}[T_n^2] &\leq& \mathbb{E}\left[\sum_{k=1}^n w_k^2(x) \mathbb{E}(T_{k,m}^2\mid X_{k,m})\right] +2\mathbb{E}\left[\sum_{i<j} w_i(x)w_j(x) \mathbb{E}(T_{i,m}T_{j,m}\mid \sigma (X_{i,m},X_{j,m}))\right]\\\notag
		&\le& C\mathbb{E}\left[\sum_{k=1}^n w_k^2(x) \right] +2\mathbb{E}\left[\sum_{i<j} w_i(x)w_j(x) \mathbb{E}(T_{i,m}\mid X_{i,m})\mathbb{E}(T_{j,m}\mid X_{j,m})\right]\\\notag
		&\le& C\mathbb{E}\left[\sum_{k=1}^n w_k^2(x) \right] \quad\text{since } \quad\mathbb{E}(T_{j,m}\mid X_{j,m})=0.
	\end{eqnarray}	
	Since \(\sum_{k=1}^n w_k(x)^2 \le \max_k w_k(x) \le s/n\), it follows that
	\[
	\mathbb{E}[T_n^2]  \le C\frac{s}{n} \to 0.
	\]	
	Therefore, \(T_n \xrightarrow{L^2} 0\) and consequently \(T_n \xrightarrow{\mathbb{P}} 0\).
	
\end{proof}
\begin{proof}{\textbf{(Proof of lemma \ref{lemma4})}}
	We have:  
	\[
	\sum_{k=1}^n w_k(x) Y_{k,m} - G_{\xi_0(x)}[\ell_{\theta_0}] =  \underbrace{\sum_{k=1}^n w_k(x)\left( Y_{k,m} - \mathbb{E}[Y_{k,m}\mid X_{k,m}]\right)}_{T_n} + \underbrace{\sum_{k=1}^n w_k(x) \bigl( \mathbb{E}[Y_{k,m}\mid X_{k,m}] - G_{\xi_0(x)}[\ell_{\theta_0}]\bigr) }_{V_n},
	\]  
	By Lemma \ref{lemma6}, we have \(T_n \xrightarrow{\mathbb{P}} 0\). It remains to show the convergence in probability of \(V_n\) to 0.\\	
	To this end, we use the local uniform convergence of  
	\( f_m(x) = \mathbb{E}[Y_{k,m} \mid X_{k,m} = x] \) toward the continuous function  
	\( \phi(x) := G_{\xi_0(x)}[\ell_{\theta_0}] \).	
	
	First, we establish the pointwise convergence of \( f_m(\cdot) \) to \( \phi(\cdot) \). For each fixed \( x \in \mathcal{X} \), Assumption \ref{assump1} ensures that the conditional distribution of \( (Z_{k,m} - b_m(x))/a_m(x) \) given \( X_{k,m} = x \) converges to \( G_{\xi_0(x)} \). By the continuous mapping theorem, we obtain the convergence in distribution of \( Y_{k,m} \) given \( X_{k,m} = x' \) toward \( \ell_{\theta_0}(Z) \), where \( Z \sim G_{\xi_0(x')} \). Moreover, equation (\ref{Hypothese_H}) provides a uniform \( L^2 \) bound, hence the sequence is uniformly integrable. Consequently,
	\begin{equation}\label{conv_ponctuel}
		f_m(x) := \mathbb{E}[Y_{k,m} \mid X_{k,m} = x] \longrightarrow \phi(x), \quad \forall x \in \mathcal{X}. 
	\end{equation}
	We now prove uniform convergence in a neighborhood of \( x \). The limit function \( \phi \) is continuous on the compact set \( \mathcal{X} \). Let \( \varepsilon > 0 \). By continuity of \( \phi \) at \( x \), there exists \( \delta_1 > 0 \) such that for all \( x' \) with \( \|x' - x\| < \delta_1 \),
	\begin{equation}\label{continut_phi}
		|\phi(x') - \phi(x)| < \frac{\varepsilon}{2}.
	\end{equation}
	On the other hand, the pointwise convergence \eqref{conv_ponctuel} together with the equicontinuity of \( \{f_m\} \)  implies uniform convergence on every compact set. In particular, on the closed ball, we have
	\[
	\sup_{\|x' - x\|_2 \leq \delta_1} |f_m(x') - \phi(x')| \longrightarrow 0 \quad \text{as } m \to \infty.
	\]
	Thus, there exists \( N_1 \) such that for all \( n \geq N_1 \) (and hence \( m = m(n) \) sufficiently large),
	\begin{equation}\label{conv_unif_compact}
		\sup_{\|x' - x\| \leq \delta_1} |f_m(x') - \phi(x')| < \frac{\varepsilon}{2}.
	\end{equation}
	Combining \eqref{continut_phi} and \eqref{conv_unif_compact}, for all \( n \geq N_1 \) and all \( x' \) such that \( \|x' - x\|_2 < \delta_1 \), we obtain
	\[
	|f_m(x') - \phi(x)| 
	\leq |f_m(x') - \phi(x')| + |\phi(x') - \phi(x)| 
	< \varepsilon.
	\]
	Setting \( \delta = \delta_1 \), we obtain the following property: for every \( \varepsilon > 0 \), there exist \( \delta > 0 \) and \( N \) such that for all \( n \geq N \) and all \( x' \) satisfying \( \|x' - x\| < \delta \),
	\[
	|f_m(x') - \phi(x)| < \varepsilon,
	\]
	that is,
	\begin{equation}\label{conv_unif_voisi_x}
		\Big|\mathbb{E}[Y_{k,m} \mid X_{k,m} = x'] - G_{\xi_0(x)}[\ell_{\theta_0}]\Big| < \varepsilon.
	\end{equation}
	Define \( A_n = \{k : \|X_{k,m} - x\|_2 \geq \delta\} \). We have the following decomposition
	\[
	\begin{aligned}
		|V_n| 
		&\leq \sum_{k \notin A_n} w_k(x) \Big|\mathbb{E}[Y_{k,m} \mid X_{k,m}] - G_{\xi_0(x)}[\ell_{\theta_0}]\Big|   + \sum_{k \in A_n} w_k(x) \Big|\mathbb{E}[Y_{k,m} \mid X_{k,m}] - G_{\xi_0(x)}[\ell_{\theta_0}]\Big|.
	\end{aligned}
	\]	
	The first term is bounded by \( \varepsilon \). Indeed,
	\begin{align}
		\sum_{k \notin A_n} w_k(x) 
		\Big|\mathbb{E}[Y_{k,m} \mid X_{k,m}] - G_{\xi_0(x)}[\ell_{\theta_0}]\Big|
		&\leq \varepsilon \sum_{k \notin A_n} w_k(x)\quad \text{ by equation \eqref{conv_unif_voisi_x}} \notag \\
		&\leq \varepsilon,\notag
	\end{align}
	since \( \sum_{k \notin A_n} w_k(x) \leq 1 \).
	For the second term, we use the following uniform bound provided by the Cauchy--Schwarz inequality and equation (\ref{Hypothese_H}):
	\begin{align}
		\Big|\mathbb{E}[Y_{k,m} \mid X_{k,m}]\Big|
		&\leq \sqrt{\mathbb{E}[Y_{k,m}^2 \mid X_{k,m}]} \notag \\
		&\leq \sqrt{C}.\notag
	\end{align}
	Similarly, the function \( x \mapsto |G_{\xi_0(x)}[\ell_{\theta_0}]| \) is bounded on the compact set $\mathcal{X}$. Hence, there exists a constant \( M \) such that
	\[
	\Big|\mathbb{E}[Y_{k,m} \mid X_{k,m}] - G_{\xi_0(x)}[\ell_{\theta_0}]\Big| \leq M.
	\]
	Therefore,
	\begin{align}
		\sum_{k \in A_n} w_k(x) 
		\Big|\mathbb{E}[Y_{k,m} \mid X_{k,m}] - G_{\xi_0(x)}[\ell_{\theta_0}]\Big|
		\leq M \sum_{k \in A_n} w_k(x).\notag
	\end{align}
	Consequently,
	\[
	|V_n| \leq \varepsilon + M \sum_{k \in A_n} w_k(x).
	\]	
	Moreover, for every \( k \in A_n \), we have  
	\[
	1 \leq \frac{1}{\delta} \|X_{k,m} - x\|_2,
	\]
	hence
	\[
	\sum_{k \in A_n} w_k(x)
	\leq \frac{1}{\delta} \sum_{k=1}^n w_k(x)\|X_{k,m} - x\|_2.
	\]
	Thus,
	\begin{align}
		\mathbb{P}\left(|V_n| > 2\varepsilon\right)
		&\leq \mathbb{P}\left(\varepsilon < M \sum_{k \in A_n} w_k(x)\right) \notag \\
		&\leq \mathbb{P}\left(\varepsilon < \frac{M}{\delta} 
		\sum_{k=1}^n w_k(x)\|X_{k,m} - x\|_2\right) \notag \\
		&\to 0,\notag
	\end{align}
	since 
	\[
	\sum_{k=1}^n w_k(x)\|X_{k,m} - x\|_2 
	\xrightarrow{\mathbb{P}} 0
	\]
	by Lemma \ref{lemme_sum_Wk_Liptchi_cond}.	
	We conclude that
	$
	V_n \xrightarrow{\mathbb{P}} 0.
	$
	In summary,
	\[
	\sum_{k=1}^n w_k(x) Y_{k,m} 
	- G_{\xi_0(x)}[\ell_{\theta_0}]
	= T_n + V_n,
	\]
	with \( V_n \xrightarrow{\mathbb{P}} 0 \) and \( T_n \xrightarrow{\mathbb{P}} 0 \). Therefore,
	\[
	\mathbb{P}_n^x[\ell_{\theta_0}] 
	\xrightarrow{\mathbb{P}} 
	G_{\xi_0(x)}[\ell_{\theta_0}].
	\]
	
\end{proof}

\section{Diagnostic Test Results for Sector Portfolios}\label{Appendix_diagnostic}

\begin{table}[h]
	\centering
	\renewcommand{\arraystretch}{1.5}
	
	\begin{tabular}{|l|c|c||c|c|}
		\hline
	\multirow{2}{*}{	\textbf{Portfolio}}
		&
		\multicolumn{2}{c||}{\textbf{Ljung-Box Test on the negative log-returns}}
		&
		\multicolumn{2}{c|}{\makecell{\textbf{ Ljung–Box test on the residuals of}\\\textbf{ the ARMA(2,2)–GARCH(1,1) model}}}
		\\ 	
		\hhline{~|----|}
		&
		\textbf{Statistic}
		&
		\textbf{p-value}
		&
		\textbf{Statistic}
		&
		\textbf{p-value}
		\\
	\hline
		NoDur & $734.1349$    & $ 10^{-16}$ 
		&1.459903 & $ 0.6915552$
		\\
	
			Durbl &$ 674.8170$ & $7.90 \times 10^{-16}$
		&$3.390246$ &  $0.3352782$
		\\
	
		Manuf & $697.8783$ & $10^{-16}$ 
		&$1.885603$ & $ 0.5964862$
		\\
	
		Enrgy & $319.3888$ & $1.11 \times 10^{-14}$
		&$2.446144$  & $0.4851027$
		\\
	
	HiTec &$634.4345$  & $4.44 \times 10^{-16}$
		&$4.687644$ &  $0.1961512$
		\\
	Telcm & $209.1109$& $6.91 \times 10^{-10}$ 
		&
		$4.155771$ &  $0.2451275$
		\\
		Shops &$797.1680 $ & $ 10^{-16}$ 
		&$4.541575$ &  $0.2086117$
		\\
	Hlth & $452.2590$ & $10^{-15}$
		&$2.204502$ &   $0.5310623$
		\\

	Utils & $113.6355$ & $ 10^{-9}$ 
		&
$	1.688870$ &  $0.6394110$
		\\
	
	Other & $735.8855 $ & $ 10^{-16}$
		&
	$2.454387$ &   $0.4835908$
		\\
		\hline
	\end{tabular}
	
	\caption{Ljung–Box test results for marginal losses and standardized residuals}
	\label{tab:diagnostic_tests}
	
\end{table}

\bibliographystyle{apalike}
 \bibliography{reference} 	


\begin{thebibliography}{41}
\ifx \bisbn   \undefined \def \bisbn  #1{ISBN #1}\fi
\ifx \binits  \undefined \def \binits#1{#1}\fi
\ifx \bauthor  \undefined \def \bauthor#1{#1}\fi
\ifx \batitle  \undefined \def \batitle#1{#1}\fi
\ifx \bjtitle  \undefined \def \bjtitle#1{#1}\fi
\ifx \bvolume  \undefined \def \bvolume#1{\textbf{#1}}\fi
\ifx \byear  \undefined \def \byear#1{#1}\fi
\ifx \bissue  \undefined \def \bissue#1{#1}\fi
\ifx \bfpage  \undefined \def \bfpage#1{#1}\fi
\ifx \blpage  \undefined \def \blpage #1{#1}\fi
\ifx \burl  \undefined \def \burl#1{\textsf{#1}}\fi
\ifx \doiurl  \undefined \def \doiurl#1{\url{https://doi.org/#1}}\fi
\ifx \betal  \undefined \def \betal{\textit{et al.}}\fi
\ifx \binstitute  \undefined \def \binstitute#1{#1}\fi
\ifx \binstitutionaled  \undefined \def \binstitutionaled#1{#1}\fi
\ifx \bctitle  \undefined \def \bctitle#1{#1}\fi
\ifx \beditor  \undefined \def \beditor#1{#1}\fi
\ifx \bpublisher  \undefined \def \bpublisher#1{#1}\fi
\ifx \bbtitle  \undefined \def \bbtitle#1{#1}\fi
\ifx \bedition  \undefined \def \bedition#1{#1}\fi
\ifx \bseriesno  \undefined \def \bseriesno#1{#1}\fi
\ifx \blocation  \undefined \def \blocation#1{#1}\fi
\ifx \bsertitle  \undefined \def \bsertitle#1{#1}\fi
\ifx \bsnm \undefined \def \bsnm#1{#1}\fi
\ifx \bsuffix \undefined \def \bsuffix#1{#1}\fi
\ifx \bparticle \undefined \def \bparticle#1{#1}\fi
\ifx \barticle \undefined \def \barticle#1{#1}\fi
\bibcommenthead
\ifx \bconfdate \undefined \def \bconfdate #1{#1}\fi
\ifx \botherref \undefined \def \botherref #1{#1}\fi
\ifx \url \undefined \def \url#1{\textsf{#1}}\fi
\ifx \bchapter \undefined \def \bchapter#1{#1}\fi
\ifx \bbook \undefined \def \bbook#1{#1}\fi
\ifx \bcomment \undefined \def \bcomment#1{#1}\fi
\ifx \oauthor \undefined \def \oauthor#1{#1}\fi
\ifx \citeauthoryear \undefined \def \citeauthoryear#1{#1}\fi
\ifx \endbibitem  \undefined \def \endbibitem {}\fi
\ifx \bconflocation  \undefined \def \bconflocation#1{#1}\fi
\ifx \arxivurl  \undefined \def \arxivurl#1{\textsf{#1}}\fi
\csname PreBibitemsHook\endcsname

\bibitem[\protect\citeauthoryear{Athey et~al.}{2019}]{athey_generalized_2019}
\begin{barticle}
\bauthor{\bsnm{Athey}, \binits{S.}},
\bauthor{\bsnm{Tibshirani}, \binits{J.}},
\bauthor{\bsnm{Wager}, \binits{S.}}:
\batitle{Generalized random forests}.
\bjtitle{The Annals of Statistics}
\bvolume{47}(\bissue{2}),
\bfpage{1148}--\blpage{1178}
(\byear{2019})
\doiurl{10.1214/18-AOS1709}
\end{barticle}
\endbibitem

\bibitem[\protect\citeauthoryear{Boudrissa
  et~al.}{2017}]{boudrissa_modelling_2017}
\begin{barticle}
\bauthor{\bsnm{Boudrissa}, \binits{N.}},
\bauthor{\bsnm{Cheraitia}, \binits{H.}},
\bauthor{\bsnm{Halimi}, \binits{L.}}:
\batitle{Modelling maximum daily yearly rainfall in northern {Algeria} using
  generalized extreme value distributions from 1936 to 2009}.
\bjtitle{Meteorological Applications}
\bvolume{24}(\bissue{1}),
\bfpage{114}--\blpage{119}
(\byear{2017})
\doiurl{10.1002/met.1610}
\end{barticle}
\endbibitem

\bibitem[\protect\citeauthoryear{Beyerlein}{2014}]{beyerlein_quantile_2014-1}
\begin{barticle}
\bauthor{\bsnm{Beyerlein}, \binits{A.}}:
\batitle{Quantile {Regression}--{Opportunities} and {Challenges} {From} a
  {User}'s {Perspective}}.
\bjtitle{American Journal of Epidemiology}
\bvolume{180}(\bissue{3}),
\bfpage{330}--\blpage{331}
(\byear{2014})
\doiurl{10.1093/aje/kwu178}
\end{barticle}
\endbibitem

\bibitem[\protect\citeauthoryear{Breiman}{2001}]{breiman_random_2001}
\begin{barticle}
\bauthor{\bsnm{Breiman}, \binits{L.}}:
\batitle{Random forests}.
\bjtitle{Machine learning}
\bvolume{45},
\bfpage{5}--\blpage{32}
(\byear{2001})
\end{barticle}
\endbibitem

\bibitem[\protect\citeauthoryear{Bremnes}{2004}]{bremnes_probabilistic_2004}
\begin{barticle}
\bauthor{\bsnm{Bremnes}, \binits{J.B.}}:
\batitle{Probabilistic {Forecasts} of {Precipitation} in {Terms} of {Quantiles}
  {Using} {NWP} {Model} {Output}}.
\bjtitle{Monthly Weather Review}
\bvolume{132}(\bissue{1}),
\bfpage{338}--\blpage{347}
(\byear{2004})
\doiurl{10.1175/1520-0493(2004)132<0338:PFOPIT>2.0.CO;2}
\end{barticle}
\endbibitem

\bibitem[\protect\citeauthoryear{Buhai}{2005}]{buhai_quantile_2005}
\begin{barticle}
\bauthor{\bsnm{Buhai}, \binits{S.}}:
\batitle{Quantile regression: overview and selected applications}.
\bjtitle{Ad Astra}
\bvolume{4}(\bissue{4}),
\bfpage{1}--\blpage{17}
(\byear{2005})
\end{barticle}
\endbibitem

\bibitem[\protect\citeauthoryear{Chernozhukov and
  Fernandez-Val}{2011}]{chernozhukov_inference_2011}
\begin{barticle}
\bauthor{\bsnm{Chernozhukov}, \binits{V.}},
\bauthor{\bsnm{Fernandez-Val}, \binits{I.}}:
\batitle{Inference for {Extremal} {Conditional} {Quantile} {Models}, with an
  {Application} to {Market} and {Birthweight} {Risks}}.
\bjtitle{The Review of Economic Studies}
\bvolume{78}(\bissue{2}),
\bfpage{559}--\blpage{589}
(\byear{2011})
\doiurl{10.1093/restud/rdq020}
\end{barticle}
\endbibitem

\bibitem[\protect\citeauthoryear{Calabrese and
  Giudici}{2015}]{calabrese_estimating_2015}
\begin{barticle}
\bauthor{\bsnm{Calabrese}, \binits{R.}},
\bauthor{\bsnm{Giudici}, \binits{P.}}:
\batitle{Estimating bank default with generalised extreme value regression
  models}.
\bjtitle{The Journal of the Operational Research Society}
\bvolume{66}(\bissue{11}),
\bfpage{1783}--\blpage{1792}
(\byear{2015})
\end{barticle}
\endbibitem

\bibitem[\protect\citeauthoryear{Chernozhukov}{2005}]{chernozhukov_extremal_2005}
\begin{botherref}
\oauthor{\bsnm{Chernozhukov}, \binits{V.}}:
Extremal quantile regression.
Annals of Statistics,
806--839
(2005)
\doiurl{10.1214/009053604000001165}
\end{botherref}
\endbibitem

\bibitem[\protect\citeauthoryear{Chaudhuri and
  Loh}{2002}]{chaudhuri_nonparametric_2002}
\begin{botherref}
\oauthor{\bsnm{Chaudhuri}, \binits{P.}},
\oauthor{\bsnm{Loh}, \binits{W.-Y.}}:
Nonparametric estimation of conditional quantiles using quantile regression
  trees.
Bernoulli,
561--576
(2002)
\end{botherref}
\endbibitem

\bibitem[\protect\citeauthoryear{Coles}{2001}]{coles_introduction_2001-1}
\begin{bbook}
\bauthor{\bsnm{Coles}, \binits{S.}}:
\bbtitle{An {Introduction} to {Statistical} {Modeling} of {Extreme} {Values}}.
\bsertitle{Springer {Series} in {Statistics}}.
\bpublisher{Springer},
\blocation{London}
(\byear{2001}).
\doiurl{10.1007/978-1-4471-3675-0}
\end{bbook}
\endbibitem

\bibitem[\protect\citeauthoryear{Dombry and
  Ferreira}{2019}]{dombry_maximum_2019}
\begin{barticle}
\bauthor{\bsnm{Dombry}, \binits{C.}},
\bauthor{\bsnm{Ferreira}, \binits{A.}}:
\batitle{Maximum likelihood estimators based on the block maxima method}.
\bjtitle{Bernoulli}
\bvolume{25}(\bissue{3}),
\bfpage{1690}--\blpage{1723}
(\byear{2019})
\doiurl{10.3150/18-BEJ1032}
\end{barticle}
\endbibitem

\bibitem[\protect\citeauthoryear{Daouia et~al.}{2013}]{daouia_kernel_2013-1}
\begin{barticle}
\bauthor{\bsnm{Daouia}, \binits{A.}},
\bauthor{\bsnm{Gardes}, \binits{L.}},
\bauthor{\bsnm{Girard}, \binits{S.}}:
\batitle{On kernel smoothing for extremal quantile regression}.
\bjtitle{Bernoulli}
\bvolume{19}(\bissue{5B}),
\bfpage{2557}--\blpage{2589}
(\byear{2013})
\doiurl{10.3150/12-BEJ466}
\end{barticle}
\endbibitem

\bibitem[\protect\citeauthoryear{Dombry}{2015}]{dombry_existence_2015}
\begin{barticle}
\bauthor{\bsnm{Dombry}, \binits{C.}}:
\batitle{Existence and consistency of the maximum likelihood estimators for the
  extreme value index within the block maxima framework}.
\bjtitle{Bernoulli}
\bvolume{21}(\bissue{1}),
\bfpage{420}--\blpage{436}
(\byear{2015})
\doiurl{10.3150/13-BEJ573}
\end{barticle}
\endbibitem

\bibitem[\protect\citeauthoryear{Farkas et~al.}{2024}]{farkas_generalized_2024}
\begin{barticle}
\bauthor{\bsnm{Farkas}, \binits{S.}},
\bauthor{\bsnm{Heranval}, \binits{A.}},
\bauthor{\bsnm{Lopez}, \binits{O.}},
\bauthor{\bsnm{Thomas}, \binits{M.}}:
\batitle{Generalized pareto regression trees for extreme event analysis}.
\bjtitle{Extremes}
(\byear{2024})
\doiurl{10.1007/s10687-024-00485-1}
\end{barticle}
\endbibitem

\bibitem[\protect\citeauthoryear{Fakoor et~al.}{2023}]{fakoor_flexible_2023}
\begin{barticle}
\bauthor{\bsnm{Fakoor}, \binits{R.}},
\bauthor{\bsnm{Kim}, \binits{T.}},
\bauthor{\bsnm{Mueller}, \binits{J.}},
\bauthor{\bsnm{Smola}, \binits{A.J.}},
\bauthor{\bsnm{Tibshirani}, \binits{R.J.}}:
\batitle{Flexible model aggregation for quantile regression}.
\bjtitle{Journal of Machine Learning Research}
\bvolume{24}(\bissue{162}),
\bfpage{1}--\blpage{45}
(\byear{2023})
\end{barticle}
\endbibitem

\bibitem[\protect\citeauthoryear{Fisher and
  Tippett}{1928}]{fisher_limiting_1928}
\begin{barticle}
\bauthor{\bsnm{Fisher}, \binits{R.A.}},
\bauthor{\bsnm{Tippett}, \binits{L.H.C.}}:
\batitle{Limiting forms of the frequency distribution of the largest or
  smallest member of a sample}.
\bjtitle{Mathematical Proceedings of the Cambridge Philosophical Society}
\bvolume{24}(\bissue{2}),
\bfpage{180}--\blpage{190}
(\byear{1928})
\doiurl{10.1017/S0305004100015681}
\end{barticle}
\endbibitem

\bibitem[\protect\citeauthoryear{Gnedenko}{1943}]{gnedenko_sur_1943}
\begin{barticle}
\bauthor{\bsnm{Gnedenko}, \binits{B.}}:
\batitle{Sur {La} {Distribution} {Limite} {Du} {Terme} {Maximum} {D}'{Une}
  {Serie} {Aleatoire}}.
\bjtitle{Annals of Mathematics}
\bvolume{44}(\bissue{3}),
\bfpage{423}--\blpage{453}
(\byear{1943})
\doiurl{10.2307/1968974}
\end{barticle}
\endbibitem

\bibitem[\protect\citeauthoryear{Gardes and
  Stupfler}{2015}]{gardes_estimating_2015}
\begin{barticle}
\bauthor{\bsnm{Gardes}, \binits{L.}},
\bauthor{\bsnm{Stupfler}, \binits{G.}}:
\batitle{Estimating extreme quantiles under random truncation}.
\bjtitle{TEST}
\bvolume{24}(\bissue{2}),
\bfpage{207}--\blpage{227}
(\byear{2015})
\doiurl{10.1007/s11749-014-0403-5}
\end{barticle}
\endbibitem

\bibitem[\protect\citeauthoryear{Gardes and
  Stupfler}{2019}]{gardes_integrated_2019}
\begin{barticle}
\bauthor{\bsnm{Gardes}, \binits{L.}},
\bauthor{\bsnm{Stupfler}, \binits{G.}}:
\batitle{An integrated functional {Weissman} estimator for conditional extreme
  quantiles}.
\bjtitle{REVSTAT-Statistical Journal}
\bvolume{17}(\bissue{1}),
\bfpage{109}--\blpage{144}
(\byear{2019})
\end{barticle}
\endbibitem

\bibitem[\protect\citeauthoryear{Gnecco et~al.}{2024}]{gnecco_extremal_2024}
\begin{botherref}
\oauthor{\bsnm{Gnecco}, \binits{N.}},
\oauthor{\bsnm{Terefe}, \binits{E.M.}},
\oauthor{\bsnm{Engelke}, \binits{S.}}:
Extremal {Random} {Forests}.
Journal of the American Statistical Association,
1--24
(2024)
\doiurl{10.1080/01621459.2023.2300522}
\end{botherref}
\endbibitem

\bibitem[\protect\citeauthoryear{Haan and
  Ferreira}{2006}]{de_haan_extreme_2006}
\begin{bbook}
\bauthor{\bsnm{Haan}, \binits{L.}},
\bauthor{\bsnm{Ferreira}, \binits{A.}}:
\bbtitle{Extreme {Value} {Theory}}.
\bsertitle{Springer {Series} in {Operations} {Research} and {Financial}
  {Engineering}}.
\bpublisher{Springer},
\blocation{New York, NY}
(\byear{2006}).
\doiurl{10.1007/0-387-34471-3}
\end{bbook}
\endbibitem

\bibitem[\protect\citeauthoryear{Hallock and
  Koenker}{2001}]{hallock_quantile_2001}
\begin{barticle}
\bauthor{\bsnm{Hallock}, \binits{K.F.}},
\bauthor{\bsnm{Koenker}, \binits{R.}}:
\batitle{Quantile regression}.
\bjtitle{The Journal of Economic Perspectives}
\bvolume{15}(\bissue{4}),
\bfpage{143}
(\byear{2001})
\end{barticle}
\endbibitem

\bibitem[\protect\citeauthoryear{Kallenberg}{2002}]{kallenberg_foundations_2002}
\begin{bbook}
\bauthor{\bsnm{Kallenberg}, \binits{O.}}:
\bbtitle{Foundations of {Modern} {Probability}}.
\bsertitle{Probability and {Its} {Applications}}.
\bpublisher{Springer},
\blocation{New York, NY}
(\byear{2002}).
\doiurl{10.1007/978-1-4757-4015-8}
\end{bbook}
\endbibitem

\bibitem[\protect\citeauthoryear{Koenker and
  Bassett}{1978}]{koenker_regression_1978}
\begin{barticle}
\bauthor{\bsnm{Koenker}, \binits{R.}},
\bauthor{\bsnm{Bassett}, \binits{G.}}:
\batitle{Regression {Quantiles}}.
\bjtitle{Econometrica}
\bvolume{46}(\bissue{1}),
\bfpage{33}
(\byear{1978})
\doiurl{10.2307/1913643}
\end{barticle}
\endbibitem

\bibitem[\protect\citeauthoryear{Koenker
  et~al.}{2017}]{chernozhukov_extremal_2017}
\begin{bbook}
\bauthor{\bsnm{Koenker}, \binits{R.}},
\bauthor{\bsnm{Chernozhukov}, \binits{V.}},
\bauthor{\bsnm{He}, \binits{X.}},
\bauthor{\bsnm{Peng}, \binits{L.}}:
\bbtitle{Handbook of {Quantile} {Regression}}.
\bpublisher{Chapman and Hall/CRC},
\blocation{New York}
(\byear{2017}).
\doiurl{10.1201/9781315120256}
\end{bbook}
\endbibitem

\bibitem[\protect\citeauthoryear{Lin and Jeon}{2006}]{lin_random_2006}
\begin{barticle}
\bauthor{\bsnm{Lin}, \binits{Y.}},
\bauthor{\bsnm{Jeon}, \binits{Y.}}:
\batitle{Random forests and adaptive nearest neighbors}.
\bjtitle{Journal of the American Statistical Association}
\bvolume{101}(\bissue{474}),
\bfpage{578}--\blpage{590}
(\byear{2006})
\end{barticle}
\endbibitem

\bibitem[\protect\citeauthoryear{Meinshausen and
  Ridgeway}{2006}]{meinshausen_quantile_2006}
\begin{botherref}
\oauthor{\bsnm{Meinshausen}, \binits{N.}},
\oauthor{\bsnm{Ridgeway}, \binits{G.}}:
Quantile regression forests.
Journal of machine learning research
\textbf{7}(6)
(2006)
\end{botherref}
\endbibitem

\bibitem[\protect\citeauthoryear{Pasche and Engelke}{2024}]{pasche_neural_2024}
\begin{barticle}
\bauthor{\bsnm{Pasche}, \binits{O.C.}},
\bauthor{\bsnm{Engelke}, \binits{S.}}:
\batitle{Neural networks for extreme quantile regression with an application to
  forecasting of flood risk}.
\bjtitle{The Annals of Applied Statistics}
\bvolume{18}(\bissue{4}),
\bfpage{2818}--\blpage{2839}
(\byear{2024})
\doiurl{10.1214/24-AOAS1907}
\end{barticle}
\endbibitem

\bibitem[\protect\citeauthoryear{Smith}{1985}]{smith_maximum_1985}
\begin{barticle}
\bauthor{\bsnm{Smith}, \binits{R.L.}}:
\batitle{Maximum likelihood estimation in a class of nonregular cases}.
\bjtitle{Biometrika}
\bvolume{72}(\bissue{1}),
\bfpage{67}--\blpage{90}
(\byear{1985})
\doiurl{10.1093/biomet/72.1.67}
\end{barticle}
\endbibitem

\bibitem[\protect\citeauthoryear{Vidagbandji
  et~al.}{2025}]{vidagbandji_generalized_2025}
\begin{barticle}
\bauthor{\bsnm{Vidagbandji}, \binits{L.M.}},
\bauthor{\bsnm{Berred}, \binits{A.}},
\bauthor{\bsnm{Bertelle}, \binits{C.}},
\bauthor{\bsnm{Amanton}, \binits{L.}}:
\batitle{Generalized random forest for extreme quantile regression}.
\bjtitle{Communications in Statistics - Simulation and Computation}
\bvolume{0}(\bissue{0}),
\bfpage{1}--\blpage{24}
(\byear{2025})
\doiurl{10.1080/03610918.2025.2543854}
\end{barticle}
\endbibitem

\bibitem[\protect\citeauthoryear{Vidagbandji
  et~al.}{2026}]{vidagbandji_penalized_2026-1}
\begin{barticle}
\bauthor{\bsnm{Vidagbandji}, \binits{L.M.}},
\bauthor{\bsnm{Berred}, \binits{A.}},
\bauthor{\bsnm{Bertelle}, \binits{C.}},
\bauthor{\bsnm{Amanton}, \binits{L.}}:
\batitle{Penalized {Estimation} of {GEV} {Parameters} for {Extreme} {Quantile}
  {Regression}}.
\bjtitle{Journal of Statistical Theory and Practice}
\bvolume{20}(\bissue{2}),
\bfpage{55}
(\byear{2026})
\doiurl{10.1007/s42519-026-00568-1}
\end{barticle}
\endbibitem

\bibitem[\protect\citeauthoryear{Velthoen
  et~al.}{2023}]{velthoen_gradient_2023}
\begin{barticle}
\bauthor{\bsnm{Velthoen}, \binits{J.}},
\bauthor{\bsnm{Dombry}, \binits{C.}},
\bauthor{\bsnm{Cai}, \binits{J.-J.}},
\bauthor{\bsnm{Engelke}, \binits{S.}}:
\batitle{Gradient boosting for extreme quantile regression}.
\bjtitle{Extremes}
\bvolume{26}(\bissue{4}),
\bfpage{639}--\blpage{667}
(\byear{2023})
\doiurl{10.1007/s10687-023-00473-x}
\end{barticle}
\endbibitem

\bibitem[\protect\citeauthoryear{van~der
  Vaart}{1998}]{van_der_vaart_asymptotic_1998}
\begin{bbook}
\bauthor{\bsnm{Vaart}, \binits{A.W.}}:
\bbtitle{Asymptotic {Statistics}}.
\bsertitle{Cambridge {Series} in {Statistical} and {Probabilistic}
  {Mathematics}}.
\bpublisher{Cambridge University Press},
\blocation{Cambridge}
(\byear{1998}).
\doiurl{10.1017/CBO9780511802256}
\end{bbook}
\endbibitem

\bibitem[\protect\citeauthoryear{Wager and Athey}{2018}]{wager_estimation_2018}
\begin{barticle}
\bauthor{\bsnm{Wager}, \binits{S.}},
\bauthor{\bsnm{Athey}, \binits{S.}}:
\batitle{Estimation and inference of heterogeneous treatment effects using
  random forests}.
\bjtitle{Journal of the American Statistical Association}
\bvolume{113}(\bissue{523}),
\bfpage{1228}--\blpage{1242}
(\byear{2018})
\end{barticle}
\endbibitem

\bibitem[\protect\citeauthoryear{Wang and Li}{2013}]{wang_estimation_2013}
\begin{barticle}
\bauthor{\bsnm{Wang}, \binits{H.J.}},
\bauthor{\bsnm{Li}, \binits{D.}}:
\batitle{Estimation of {Extreme} {Conditional} {Quantiles} {Through} {Power}
  {Transformation}}.
\bjtitle{Journal of the American Statistical Association}
\bvolume{108}(\bissue{503}),
\bfpage{1062}--\blpage{1074}
(\byear{2013})
\doiurl{10.1080/01621459.2013.820134}
\end{barticle}
\endbibitem

\bibitem[\protect\citeauthoryear{Wang et~al.}{2012}]{wang_estimation_2012}
\begin{barticle}
\bauthor{\bsnm{Wang}, \binits{H.J.}},
\bauthor{\bsnm{Li}, \binits{D.}},
\bauthor{\bsnm{He}, \binits{X.}}:
\batitle{Estimation of {High} {Conditional} {Quantiles} for {Heavy}-{Tailed}
  {Distributions}}.
\bjtitle{Journal of the American Statistical Association}
\bvolume{107}(\bissue{500}),
\bfpage{1453}--\blpage{1464}
(\byear{2012})
\doiurl{10.1080/01621459.2012.716382}
\end{barticle}
\endbibitem

\bibitem[\protect\citeauthoryear{Wang and Tsai}{2009}]{wang_tail_2009}
\begin{barticle}
\bauthor{\bsnm{Wang}, \binits{H.}},
\bauthor{\bsnm{Tsai}, \binits{C.-L.}}:
\batitle{Tail index regression}.
\bjtitle{Journal of the American Statistical Association}
\bvolume{104}(\bissue{487}),
\bfpage{1233}--\blpage{1240}
(\byear{2009})
\end{barticle}
\endbibitem

\bibitem[\protect\citeauthoryear{Youngman}{2019}]{youngman_generalized_2019}
\begin{barticle}
\bauthor{\bsnm{Youngman}, \binits{B.D.}}:
\batitle{Generalized {Additive} {Models} for {Exceedances} of {High}
  {Thresholds} {With} an {Application} to {Return} {Level} {Estimation} for
  {U}.{S}. {Wind} {Gusts}}.
\bjtitle{Journal of the American Statistical Association}
\bvolume{114}(\bissue{528}),
\bfpage{1865}--\blpage{1879}
(\byear{2019})
\doiurl{10.1080/01621459.2018.1529596}
\end{barticle}
\endbibitem

\bibitem[\protect\citeauthoryear{Youngman}{2022}]{youngman_evgam_2022}
\begin{barticle}
\bauthor{\bsnm{Youngman}, \binits{B.D.}}:
\batitle{evgam: {An} {R} package for generalized additive extreme value
  models}.
\bjtitle{Journal of Statistical Software}
\bvolume{103},
\bfpage{1}--\blpage{26}
(\byear{2022})
\end{barticle}
\endbibitem

\bibitem[\protect\citeauthoryear{Zhou}{2009}]{zhou_existence_2009}
\begin{barticle}
\bauthor{\bsnm{Zhou}, \binits{C.}}:
\batitle{Existence and consistency of the maximum likelihood estimator for the
  extreme value index}.
\bjtitle{Journal of Multivariate Analysis}
\bvolume{100}(\bissue{4}),
\bfpage{794}--\blpage{815}
(\byear{2009})
\doiurl{10.1016/j.jmva.2008.08.009}
\end{barticle}
\endbibitem

\end{thebibliography}
 \end{document}